\documentclass[draftcls,onecolumn]{IEEEtran}
\usepackage{setspace}
\usepackage[margin=1in]{geometry}
\usepackage{color}
\usepackage{csquotes}
\usepackage{enumitem}
\usepackage{float}
\usepackage{multirow}

\floatstyle{ruled}
\newfloat{algorithm}{tbp}{loa}
\providecommand{\algorithmname}{Algorithm}
\floatname{algorithm}{\protect\algorithmname}
\usepackage{hyperref}
\usepackage{amsopn}

\makeatletter
\newcommand{\manuallabel}[2]{\def\@currentlabel{#2}\label{#1}}
\makeatother
\usepackage{amsmath}
\usepackage{amssymb}
\usepackage{amsthm}
\usepackage{bbm}
\usepackage{tikz}
\usepackage{mathrsfs}
\usepackage{pgfplots}
\pgfplotsset{compat=1.14}
\usetikzlibrary{arrows}
\usetikzlibrary{arrows.meta}

\newtheorem{theorem}{Theorem}
\newtheorem{lemma}[theorem]{Lemma}

\newtheorem{dfn}[theorem]{Definition}

\newtheorem{cor}[theorem]{Corollary}

\newcommand{\edag}{\mathcal{E}^{\dagger}}
\newcommand{\p}{\mathbb{P}}
\newcommand{\e}{\mathbb{E}}

\newcommand{\erate}{E}

\newcommand{\calO}{\mathcal{O}}

\newcommand{\calX}{\mathcal{X}}
\newcommand{\calY}{\mathcal{Y}}
\newcommand{\calC}{\mathcal{C}}
\newcommand{\calP}{\mathcal{P}}

\newcommand{\rate}{\mathcal{E}}

\newcommand{\rateone}{\mathcal{E}^{(1)}}
\newcommand{\db}{d_{\rm B}}

\newcommand{\dc}{d_{\rm C}}

\newcommand{\feedbackrate}{\rate^f}
\newcommand{\pe}{p_{\rm err}}
\title{Maxflow-Based Bounds for Low-Rate \\ Information Propagation over Noisy Networks}
\author{Yan Hao Ling and Jonathan Scarlett}

\begin{document}

\maketitle

\begin{abstract}
    We study error exponents for the problem of low-rate communication over a directed graph, where each edge in the graph represents a noisy communication channel, and there is a single source and destination.  We derive maxflow-based achievability and converse bounds on the error exponent that match when there are two messages and all channels satisfy a symmetry condition called pairwise reversibility.  More generally, we show that the upper and lower bounds match to within a factor of 4.  We also show that with three messages there are cases where the maxflow-based error exponent is strictly suboptimal, thus showing that our tightness result cannot be extended beyond two messages without further assumptions.
\end{abstract}

\long\def\symbolfootnote[#1]#2{\begingroup\def\thefootnote{\fnsymbol{footnote}}\footnote[#1]{#2}\endgroup}

\symbolfootnote[0]{The authors are with the  Department of Computer Science and Department of Mathematics, School of Computing, National University of Singapore (NUS). Jonathan Scarlett is also with the Institute of Data  Science, NUS. Emails: \url{lingyh@nus.edu.sg};  \url{scarlett@comp.nus.edu.sg}

This work was supported by the Singapore National Research Foundation (NRF) under grant number R-252-000-A74-281.}

\section{Introduction}

Problems of low-rate (e.g., 1-bit) information propagation in noisy networks have recently gained increasing attention, with prominent examples including the following:
\begin{itemize}
    \item The transmission of a single bit over a tandem of channels (the 2-hop setting) was studied in \cite{onebit,jog2020teaching,teachlearn}, and the multi-bit variant was studied in \cite{multibit}.  The goal is to characterize the associated error exponent; as we further discuss below, this line of works is the one most relevant to the present paper.
    \item The transmission of a single bit over a long chain of noisy channels (connected by relays) dates back to the work of Schulman and Rajagopalan \cite{schulman_1994} and recently regained  attention \cite{onebit,infovelocity,domanovitz2022information} after being posed by Polyanskiy (who termed it the \emph{information velocity} problem) and connected to the 2-hop setting by Huleihel, Polyanskiy, and Shayevitz \cite{onebit}.
    \item Another line of works studied \emph{broadcasting} problems on various kinds of graphs, such as trees \cite{evans2000broadcasting}, random graphs \cite{makur2020random}, and grids \cite{makur2022grid}.  Here the goal is to reconstruct the original bit from the information received at the leaves, and the focus has been on simple intermediate nodes performing memoryless operations.  (In contrast, the problems described above allow the intermediate/relay nodes to perform complicated coding strategies.)
\end{itemize}
In this paper, we address a notable gap in this list by extending the 2-hop problem to transmission over a fixed but arbitrary directed graph, with a single source and a single destination.\footnote{General graphs also fall under the framework of \cite{schulman_1994}, even without the constraint that its size is fixed with respect to the block length, but the focus therein is only on scaling laws and not precise constants or error exponents.}  An extensive range of channel capacity results are known for such settings \cite[Ch.~15]{gamalkim}, but 1-bit and other low-rate settings have remained less explored.

Specifically, we generalize the 2-hop multi-bit setup to arbitrary networks and show a maxflow-based achievability bound (Theorem \ref{thm:main}), which supersedes our main previous multi-bit result for the 2-hop setting \cite{multibit}.  
We also provide a converse bound (Theorem \ref{thm:main_converse}) based on the cutset idea, which matches the achievability bound in the 1-bit case when every channel satisfies a symmetry condition called pairwise reversibility (see Definition \ref{dfn:pairwise_reversible} below).  While we are unable to find the optimal error exponent more generally, our achievability and converse bounds provide a 4-approximation (Lemma \ref{lem:4approx}).

We proceed by outlining the related work in more detail, and then formally describe the problem setup and state our main results.

\subsection{Related Work}

{\bf Tandem of channels (2-hop).} The problem of transmitting a \emph{single bit} over a tandem of channels was introduced by Huleihel, Polyanskiy, and Shayevitz \cite{onebit}, as well as Jog and Loh  using different motivation/terminology based on teaching and learning in multi-agent problems \cite{jog2020teaching}.  Both of these works gave various protocols with varying trade-offs between their simplicity and their associated error exponents.  We then showed in \cite{teachlearn} that the 1-hop and 2-hop error exponents coincide whenever the two channels have the same transition law.  In \cite{multibit}, we generalized this problem to the multi-bit setting and showed that whenever both channels are pairwise reversible, the 2-hop error exponent is simply given by the bottleneck (i.e., the worse of the two 1-hop error exponents).

{\bf Information velocity (many-hop).} As we outlined above, the information velocity problem seeks to establish the minimum possible ratio of the number of hops to the block length to reliably send a single bit over a long chain of relays connected by binary symmetric channels (or more generally, other channels).  This problem also has interesting connections to the 2-hop setting \cite{onebit}.

A general result of Schulman and Rajagopalan \cite{schulman_1994} translates noiseless protocols (on arbitrary graphs) to noisy ones, and when specialized to the line graph, this proves that positive information velocity is attainable.  We gave a simple and efficient protocol achieving positive information velocity in \cite{infovelocity}, though precisely characterizing the best possible information velocity remains very much open.  The information velocity of erasure-type channels was studied in \cite{domanovitz2022information}, and for this class of channels an exact characterization was attained.

{\bf Broadcasting problems.} As mentioned above, problems of broadcasting \cite{evans2000broadcasting,makur2020random,makur2022grid} have a similar flavor to the information velocity problem, but with some substantial differences including (i) the consideration of more complicated graphs, (ii) the combination of many nodes' information to form the final estimate, and (iii) the consideration of simple memoryless intermediate nodes without the use of coding.

{\bf Capacity results on graphs.} Another notable line of works has sought to characterize the channel capacity of noisy networks, as summarized in \cite[Sec.~15.10]{cover_thomas} and \cite[Ch.~15]{gamalkim}.  The study of capacity is largely quite distinct from that of 1-bit (or other low-rate) communication, but we will see some similarities throughout the paper.  In particular, our results are based on maxflow and mincut ideas, which are widely used in capacity results.  Instead of using 1-hop channel capacities as edge weights for the maxflow problem, we will use 1-hop error exponents.  Our protocols and analysis will be substantially different compared to capacity results and will require further assumptions such as pairwise reversibility, but we will use an existing idea of replacing individual edges by multiple parallel edges of smaller weight in order to form edge-disjoint paths (e.g., see \cite[Sec.~15.3]{gamalkim}).

\section{Problem Setup}

We first describe the classical channel coding setup on a discrete memoryless channel (DMC) $P$, whose alphabets are denoted by $\calX$ and $\calY$ (or sometimes $\calX_P$ and $\calY_P$). The source node is given a message $m \in \{1,2,\ldots, M\}$, and has $n$ uses of the channel $P$ to send information across to the destination node, after which the destination node guesses the value of $m$. An error is said to have occurred if the destination's guess is different from $m$. The minimum probability of an error across all protocols with $n$ time steps is denoted by $\pe(n,M,P)$, and the error exponent  of the channel with $M$ messages is defined as
\begin{equation}
\rate_{M,P} = \liminf_{n\rightarrow \infty} -\frac{1}{n} \pe(n,M,P),
\end{equation}
where the limit is as  $n\rightarrow\infty$ while $M$ remains fixed. We will use $\rate$, $\rate_M$ or $\rate_P$ whenever the omitted parameters are clear from the context.

In this paper, we consider a communication over a directed graph with edges representing noisy channels.  Specifically, we are given a directed graph $G$ (not necessarily acyclic), with two special nodes called source and destination.  Without loss of generality, we assume that there exists at least one path from the source to the destination.  For every edge $e$, there is a DMC $P_e$.  We note that for a given pair of nodes, say $(i,j)$, we allow for the possibility of multiple parallel edges connecting $i$ to $j$ (though this could equivalently be viewed as a single edge whose associated channel is a product channel).

At time 0, the source node receives a message $m \in \{1,2,\ldots, M\}$.  
For each time step, every node gets one use of the channel associated with its outgoing edge. This process occurs for a total of $n$ time steps.  After $n$ time steps, the destination node forms an estimate of the message; let $\pe(n, M, G)$ be the resulting error probability.\footnote{Our results are unaffected by whether this error probability is averaged over all $m$ or taken with respect to the worst-case $m$, as is the case for the vast majority of error exponent studies (e.g, \cite{berlekampI,berlekamp}).}  We are interested in the error exponent of the network, given by
\begin{equation}
\rate_{M, G} = \liminf_{n \to \infty} -\frac{1}{n} \log \pe(n,M,G).
\end{equation}
Note that we are taking the limit as $n\rightarrow\infty$, while $M$ and $G$ remain fixed.  When the context is clear, we will simply denote this as $\rate_G$.

For an edge $e$, we let $\rate_{e}$ be short for $\rate_{P_e}$ (i.e., the 1-hop error exponent of the associated channel).  For a fixed $(M,G)$ pair, we can create a network flow problem (see Section \ref{sec:flow}) that assigns a capacity\footnote{This is ``capacity'' in a generic sense, and not the Shannon capacity.} of $\rate_e$ on every edge $e$, and consider the maximum flow across this network. We will refer to this as the maxflow for the $(M,G)$ pair, and denote it by $|f_{M,G}|$. We will also shorten this to $|f|$ when the context is clear.

Our main result is the following achievability result for channels that are \emph{pairwise reversible}; this is a classical notion of symmetry \cite{berlekamp} that notably includes the binary symmetric channel (BSC), binary erasure channel (BEC), and their generalizations to non-binary inputs.   Briefly, the property is that the quantity $\sum_{y} P(y|x)^{1-s} P(y|x')^s$ attains its minimum at $s = \frac{1}{2}$ for all $(x,x')$, which implies that the Chernoff distance between two inputs simplifies to the Bhattacharyya distance.  See Definition \ref{dfn:pairwise_reversible} below for further details, and \cite{berlekamp} for a more comprehensive treatment with numerous examples.

\begin{theorem}
If every channel in $G$ is pairwise reversible, then the error exponent $|f_{M,G}|$ is achievable, i.e.,
\begin{equation}
\rate_{M,G} \geq |f_{M,G}|.
\end{equation}
\label{thm:main}
\end{theorem}
\begin{proof}
    We provide the proof for the series case (i.e., a line graph) in Section \ref{sec:series}, which we then use as a stepping stone for the general case in Section \ref{sec:proof_main}.
\end{proof}

We note that this is an information-theoretic achievability result, and we do not attempt to establish an associated computationally efficient protocol (though doing so would be of interest).  By comparison, our related works on 2-hop settings contained various results based on both efficient and inefficient protocols \cite{teachlearn,multibit}.

We also prove the following converse bound:
\begin{theorem}
For any pair $(M,G)$, we have
\begin{equation}
\rate_{M,G} \leq |f_{2,G}|.
\end{equation}
\label{thm:main_converse}
\end{theorem}
\begin{proof}
    See Section \ref{sec:main_converse}.
\end{proof}

Combining Theorems \ref{thm:main} and \ref{thm:main_converse}, we establish that equality holds whenever every channel is pairwise reversible and $M=2$:
\begin{cor}
If every channel in $G$ is pairwise reversible, then
\begin{equation}
\rate_{2,G} = |f_{2,G}|.
\end{equation}
\end{cor}

We are also interested in the zero-rate error exponent. Momentarily returning to classical 1-hop communication, for $R>0$, we define the error exponent of the channel $P$ at rate $R$ by 
\begin{equation}
    \erate_P(R) = \liminf_{n \rightarrow \infty} \left\{ -\frac1n\log \pe(n, e^{nR}, P) \right\},
\end{equation}
and we extend this function to $R=0$ via $\erate_P(0) = \lim_{R \to 0^+} \erate_P(R)$. Analogously, in the setting of the present paper, we let $\erate_G(R)$ and $\erate_G(0)$ denote the zero-rate error exponent of the network. We can again consider the network flow problem which assigns a capacity of $\erate_e(0)$ on every edge $e$, which we will denote by $|f_G(0)|$ or sometimes simply $|f|$.  The following result extends Theorem \ref{thm:main} to zero-rate error exponents, even without requiring the pairwise reversible assumption:
 \begin{theorem}
 For any $G$, whose channels need not be pairwise reversible, we have
 \begin{equation}
\erate_{G}(0) \geq |f_G(0)|.
\end{equation}
\label{thm:main_zero_rate}
 \end{theorem}
 \begin{proof}
    See Section \ref{sec:zero_rate_proof}. 
 \end{proof}

 Some additional results are given in Sections \ref{sec:ach_general} and \ref{sec:further}, and are briefly outlined as follows:
 \begin{itemize}
     \item Lemma \ref{lem:ach_general} states an achievability result for channels that may not be pairwise reversible; the idea is to use a weaker ``symmetrized'' error exponent on each edge in the graph, rather than the optimal 1-hop exponent.
     \item In Theorem \ref{thm:converse_other}, we provide conditions (albeit somewhat restrictive) under which a matching converse can be obtained for pairwise reversible channels with $M > 2$.
     \item In Lemma \ref{lem:4approx} we establish the optimal exponent to within a factor of $4$ for general $M$ and $G$, and in Lemma \ref{lem:2approx} we provide conditions under which this can be improved to a factor of $2$.
     \item In Theorem \ref{thm:counter}, we show that in general there exist scenarios in which the maxflow-based error exponent is strictly suboptimal when $M=3$, even under the constraint of an acyclic graph and pairwise reversible channels.  Thus, we place a strong restriction on the extent to which our tightness result for $M=2$ can be generalized.
 \end{itemize}

\section{Preliminaries} \label{sec:prelim}

In this section, we introduce some additional notation and definitions, and provide a number of useful auxiliary results that will be used for proving our main results.

\subsection{Flows and Cuts}
\label{sec:flow}

We momentarily depart from the above setup and consider generic notions on graphs.  Consider a directed graph $G=(V,E)$, with two special nodes called the source $v_s$ and destination $v_t$. Suppose that on every edge $e$, there is a non-negative real-valued capacity, denoted $c_e$.

Let $f = \{f_e\}_{e \in E}$ be a function from $E$ to $\mathbb{R}_{\geq 0}$. We say that $f$ is a flow if the following conditions hold:

\begin{itemize}
\item (Capacity constraints)
For each $e$, $0\leq f_e \leq c_e$.
\item (Flow conservation) For each $v\in V$, except for the source and destination, we have
\begin{equation}
\sum_{u:(u,v)\in E} f_{u,v} = \sum_{v:(v,w) \in E} f_{v,w}.
\end{equation}
\end{itemize}
Define the size of a flow $f$ as
\begin{equation}
|f| = \sum_{u:(u,v_s) \in E} f_{u,v_s}.
\end{equation}
A flow $f$ is maximal if for any other flow $f'$, $|f| \geq |f'|$.

Closely related to flows is the concept of cuts. A cut $c$ is a partition of $V$ into two sets $(A,B)$ such that $v_{s} \in A$ and $v_{t} \in B$. Define the size of a cut $c$ as
\begin{equation}
|c| = \sum_{u\in A, v \in B} f_{u,v},
\end{equation}
where we emphasize that this only counts edges from $A$ to $B$ and not back-edges.  A cut $c$ is minimal if for any cut $c'$, $|c| \leq |c'|$.
We now state the well-known maxflow-mincut duality theorem:
\begin{theorem}
    {\em (Maxflow-mincut duality theorem, e.g., \cite{stanford})} The size of the maximum flow is equal to the size of the minimum cut.
    \label{thm:duality}
\end{theorem}

The following result is also well known, but we provide a short proof for the interested reader.

\begin{theorem}
    {\em (Flow decomposition theorem, e.g., \cite{stanford})} Let $f$ be any flow. There exist finitely many simple paths\footnote{A path is \emph{simple} if it never visits the same node more than once.} from $v_s$ to $v_t$,  denoted by $p_1, p_2,\ldots, p_k$, with associated flows $f_1, f_2,\ldots, f_k$, such that
    \begin{itemize}
    \item Each $f_i$ flows only along path $p_i$. In other words, $(f_i)_e = |f_i|$ for all $e \in p_i$ and $(f_i)_e = 0$ otherwise.
    \item The total sum of all $f_i$ is equal to $f$. In other words, we have for all $e$ that
    \begin{equation}
    \sum_{i} (f_i)_e = f_e.
    \end{equation}
    \end{itemize}
    \label{thm:flow_decomp}
\end{theorem}

\begin{proof}
We describe a procedure for iteratively building up a list of $(p_i,f_i)$ pairs.  
Given $f$, let $p_i$ be any path that has non-zero flow across every edge. Let $|f_i|$ be the smallest among these (non-zero) values, and let $f_i$ be the flow formed by pushing $|f_i|$ units of flow along the path $p_i$. We then append $(p_i, f_i)$ to the list of pairs, subtract $|f_i|$ along every edge on $p_i$, and repeat the preceding steps until no flow remains.  Since the number of active edges (i.e., those with non-zero flow) strictly decreases on every step of this process, it must terminate and produce a finite number of paths, as desired.
\end{proof}

\subsection{Chernoff Divergence}

Let $P(y|x)$ be a generic discrete memoryless channel. 
To lighten notation, we let $P_x(\cdot)$ denote the conditional distribution $P(\cdot | x)$.  An \emph{$(M,\ell)$-codebook} is defined to be a collection of $M$ codewords each having length $\ell$, and when utilizing such a codebook, we will use the notation $(x^{(1)}, \ldots, x^{(M)})$ for the associated codewords.

For two probability distributions $Q, Q'$ over some finite set $\calX$, the Bhattacharyya distance is defined as
\begin{equation}
    \db(Q, Q') = -\log \sum_{x \in \calX} \sqrt{Q(x)Q'(x)}.
\end{equation}
For random variables $X_1, X_2$, we will also use $\db$ to refer to the Bhattacharyya distance between their distributions $P_{X_1}$ and $P_{X_2}$, i.e.
\begin{equation}
    \db(X_1, X_2) = \db(P_{X_1}, P_{X_2}), \label{eq:db_rv}
\end{equation}
and for $x, x' \in \calX_P$, we define the Bhattacharyya distance associated with two channel inputs as 
\begin{equation}
    \db(x, x', P) = \db(P_x, P_{x'}) \label{eq:dB_channel}
\end{equation}
with a slight abuse of notation.

Generalizing the Bhattacharyya distance, the Chernoff divergence with parameter $s$ is given by
\begin{equation}
    \dc(Q, Q',s) =  -\log \sum_{x \in \calX}  Q(x)^{1-s}Q'(x)^{s}, \label{eq:dc}
\end{equation}
and the Chernoff divergence (with optimized $s$) is given by
\begin{equation}
    \dc(Q, Q') = \max_{0\leq s \leq 1} \dc(Q, Q', s).
\end{equation}
Analogous to \eqref{eq:db_rv}--\eqref{eq:dB_channel}, we also write
\begin{gather}
    \dc(X_1,X_2) = \dc(P_{X_1},P_{X_2}),\\
    \dc(x, x', P) = \dc(P_x, P_{x'}). \label{eq:dC_channel}
\end{gather}

For a certain class of channels, the quantities $\dc$ and $\db$ coincide, giving us the following definition:
\begin{dfn}
{\em \cite{berlekamp}}
A discrete memoryless channel is \textbf{pairwise reversible} if, for all $x, x' \in \calX_P$,
\begin{equation}
    \db(x, x', P) = \dc(x, x', P),
\end{equation}
which occurs when the quantity
    \begin{equation}
        \sum_{y \in \calY_{P}} P(y|x)^{1-s} P(y|x')^s
    \end{equation}
attains its minimum at $s=\frac{1}{2}$.
\label{dfn:pairwise_reversible}
\end{dfn}

The pairwise reversibility assumption allows for a number of useful results and constructions that are not available for general channels.  In particular, Lemma \ref{lem:opt_rev} below gives an explicit expression for the 1-hop error exponent, and Lemma \ref{lem:db_codebook} below is based on an explicit construction that attains that exponent.  We again refer the reader to \cite{berlekamp} for a detailed discussion on pairwise reversibility with several examples.

For any positive integer $k$, we let $P^k$ denote the $k$-fold product of $P$, with probability mass function
\begin{equation}
    P^k(\vec{y}|\vec{x}) = \prod_{i=1}^k P(y_i|x_i).
\end{equation}
For two sequences $\vec{x}, \vec{x}'$ of length $k$, we also use the notation $\db(\vec{x}, \vec{x}', P^k)$ and $\dc(\vec{x}, \vec{x}', P^k)$ similarly to \eqref{eq:dB_channel} and \eqref{eq:dC_channel}, with the understanding that $\vec{x}, \vec{x}'$ are treated as inputs to $P^k$. 
We will use a well-known tensorization property of $\dc$ (and hence also $\db$), stated as follows.

\begin{lemma} 
{\em (e.g., see \cite[Lemma 8]{multibit})}
    For any sequences $\vec{x} = (x_1,\dotsc,x_k)$ and $\vec{x}' = (x'_1,\dotsc,x'_k)$, we have
    \begin{equation}
        \dc(\vec{x}, \vec{x}', P^k, s) = \sum_{i=1}^k \dc(x_i, x'_i, P, s),
        \label{eq:iid1}
    \end{equation}
    and
    \begin{equation}
        \dc(\vec{x}, \vec{x}', P^k) = \max_{0\leq s\leq 1} \sum_{i=1}^k \dc(x_i, x'_i, P, s).
        \label{eq:iid2}
    \end{equation}
    \label{lem:iid}
\end{lemma}

\subsection{1-hop Error Exponents}

We will use several results from \cite{berlekamp} on the 1-hop error exponents.  We first state two closely-related results that are implicit in \cite[Thm.~2]{berlekamp} and its proof.

\begin{lemma} \label{lem:chernoff}
    \emph{(Implicit in \cite[Thm.~2]{berlekamp})}
    Letting $\vec{y} \in \calY^n$ be the received sequence from $n$ uses of the channel $P$, we have for any fixed $M$ and any sequence of codebooks (indexed by $n$) that
    \begin{equation}
        \liminf_{n \to \infty} -\frac{1}{n} \log \pe \ge \liminf_{n \to \infty} \min_{m_1 \ne m_2} \dc((\vec{y}|m=m_1), (\vec{y}|m=m_2)).
    \end{equation}
\end{lemma}

\begin{theorem}
    \emph{(Implicit in \cite[Thm.~2]{berlekamp})}
    For any $\edag < \rateone_M$, it holds for all sufficiently large $\ell$ that there exists an $(M,\ell)$-codebook $\calC$ such that for all pairs of codewords ($\vec{x}, \vec{x}')$ in $\calC$, we have
    \begin{equation}
        \dc(\vec{x}, \vec{x}', P^\ell) \geq \ell \cdot \edag.
        \label{eq:berlekamp}
    \end{equation}
    \label{thm:berlekamp}
\end{theorem}

Next, we state explicit formulas/bounds for error exponents in the zero-rate and fixed-$M$ settings, respectively.

\begin{theorem}
    \emph{\cite[Thm.~4]{berlekamp}} For any $P$, the zero-rate error exponent is given by
    \begin{equation}
        \erate_P(0) = \max_{q \in \calP(\calX)} \sum_{x, x' \in \calX} q_{x} q_{x'} \db(x, x', P).
    \label{eq:berlekamp_zero_rate}
    \end{equation}
    \label{thm:berlekamp_zero_rate}
\end{theorem}

\begin{lemma}
    \emph{\cite[Thm.~3]{berlekamp}} For any $(M,P)$, we have
    \begin{equation}
        \rate_{M,P} \geq \frac{2}{M(M-1)} \max_{(x_1, \ldots, x_M) \in \calX^M} \sum_{1\leq m_1 < m_2 \leq M} \db(x_{m_1}, x_{m_2}, P),
        \label{eq:opt_rev}
    \end{equation}
with equality when $P$ is pairwise reversible.
\label{lem:opt_rev}
\end{lemma}
In accordance with Lemma \ref{lem:opt_rev}, we will use the following definition:
\begin{dfn}
    Define the Bhattacharyya distance based error exponent by
    \begin{equation}
        \tilde{\rate}_{M,P} = \frac{2}{M(M-1)} \max_{(x_1, \ldots, x_M) \in \calX^M} \sum_{1\leq m_1 < m_2 \leq M} \db(x_{m_1}, x_{m_2}, P), \label{eq:rev_rate}
    \end{equation}
    and note that $\tilde{\rate}_{M,P} = \rate_{M,P}$ for pairwise reversible channels.
    \label{dfn:rev_rate}
\end{dfn}

In the proof of Lemma \ref{lem:opt_rev}, \cite{berlekamp} also showed the following:
\begin{lemma}
    For any DMC $P$, there exists an $(M, \ell)$-codebook $(\vec{x}_1, \ldots, \vec{x}_M)$ with $\ell = M!$ such that for any pair $m_1 \neq m_2$, it holds that
    \begin{equation}
        \db(\vec{x}_{m_1}, \vec{x}_{m_2}, P^{\ell}) = \ell \cdot \tilde{\rate}_{M,P}.
    \end{equation}
\label{lem:db_codebook}
\end{lemma}

Since we will make use of codes over product channels, the following lemma is useful:
\begin{lemma}
Let $\ell$ be a fixed integer and $\vec{x}_1, \vec{x}_2, \ldots, \vec{x}_M$ be length-$\ell$ codewords, and let $Q$ be the restriction of $P^\ell$ to these inputs. Then $\rate_P \geq \frac{1}{\ell}\rate_Q$. Furthermore, if $P$ is pairwise reversible, then $Q$ is also pairwise reversible.
\label{lem:block_reduction}
\end{lemma}

\begin{proof}
Let $\vec{w}_1, \ldots, \vec{w}_M$ be codewords achieving an error exponent arbitrarily close to $\rate_Q$. Since each character in $\vec{w}_i$ is a string of $\ell$ characters over $X$, we can expand the strings $\vec{w}_1, \ldots, \vec{w}_M$ to form $\vec{w}_1', \ldots, \vec{w}_M'$, achieving the same error probability when used as codewords for $P$. The $\frac1\ell$ scaling factor comes from the factor of $\ell$ difference in the block length.

For the second claim, we apply Lemma \ref{lem:iid} to obtain
\begin{equation}
    \dc(\vec{x}_{m_1}, \vec{x}_{m_2}, Q) = \sum_{i=1}^\ell \dc(\vec{x}_{m_1}(i), \vec{x}_{m_2}(i), P),
\end{equation}
Since $P$ is pairwise reversible, the right-hand side attains its minimum at $s=1/2$, which means that $Q$ is also pairwise reversible.
\end{proof}

The next lemma allows us to break down the error exponent of a product channel into its individual components when $M=2$.  (We note that this result does not extend to $M > 2$.)

\begin{lemma}
    Let $P^{(1)}, P^{(2)},\ldots, P^{(k)}$ be arbitrary channels, and let $\hat{P} = P^{(1)} \times P^{(2)} \times \ldots \times P^{(k)}$ be the product channel. Then
    \begin{equation}
    \rate_{2,\hat{P}} \leq \sum_i \rate_{2,P^{(i)}}.
    \end{equation}
    \label{lem:parallel}
\end{lemma}
\begin{proof}
    For any DMC, the error exponent for $M=2$ is given as follows \cite[Eq.~(1.19)]{berlekamp}:
    \begin{equation}
    \rate_{2,P} = \max_{x_1, x_2 \in \calX_P} \dc(x_1, x_2, P).
    \label{eq:two_codewords}
    \end{equation}
    Accordingly, we choose $\vec{x}_1, \vec{x}_2 \in \prod_i {\calX_{P^{(i)}}}$ attaining the maximum in
    \begin{equation}
    \rate_{2,\hat{P}} = \max_{\vec{x}_1, \vec{x}_2} \dc(\vec{x}_1, \vec{x}_2, \hat{P}),
    \end{equation}
    and let $\vec{x}_1(i)$ and $\vec{x}_2(i)$ represent the $i$-th component of $\vec{x}_1(i)$ and $\vec{x}_2(i)$ respectively.  Then,
    \begin{eqnarray}
     \dc(\vec{x}_1, \vec{x}_2, P) & = & \max_{0\leq s \leq 1}\dc(\vec{x}_1, \vec{x}_2, P, s) \\
     & \stackrel{Lem.~\ref{lem:iid}}{=} & \max_{0\leq s \leq 1}\sum_i \dc(\vec{x}_1(i), \vec{x}_2(i), P^{(i)}, s) \\
     & \leq & \sum_i \max_{0\leq s \leq 1} \dc(\vec{x}_1(i), \vec{x}_2(i), P^{(i)}, s)\\
     & \stackrel{\eqref{eq:two_codewords}}{\leq} & \sum_{i} \rate_{2, P^{(i)}}
    \end{eqnarray}
    as required.
\end{proof}

The following result shows that we have equality in Lemma \ref{lem:parallel} if we restrict all $P^{(i)}$ to be pairwise reversible:

\begin{lemma}
    Let $P^{(1)}, P^{(2)},\ldots, P^{(k)}$ be pairwise reversible channels and let $\hat{P} = P^{(1)} \times P^{(2)} \times \ldots \times P^{(k)}$ be the product channel. Then
    \begin{equation}
    \rate_{M,\hat{P}} = \sum_i \rate_{M,P^{(i)}}.
    \end{equation}
    \label{lem:prod_parallel_reversible}
\end{lemma}
\begin{proof}
    We first note from the second part of Lemma \ref{lem:block_reduction} that $\hat{P}$ is pairwise reversible.  As a result, by Lemma \ref{lem:opt_rev}, there exists some $(\vec{x}_1, \vec{x}_2, \ldots, \vec{x}_M)$ such that
    \begin{equation}
    \rate_{\hat{P}} = \frac{2}{M(M-1)} \sum_{1\leq m_1 < m_2 \leq M} \db(\vec{x}_{m_1}, \vec{x}_{m_2}, \hat{P}).
    \end{equation}
    Moreover, by Lemma \ref{lem:iid}, we can break this  down into the Chernoff exponents for the individual channels:
    \begin{equation}
    \rate_{\hat{P}} = \frac{2}{M(M-1)} \sum_{1\leq m_1 < m_2 \leq M} \sum_i \db(\vec{x}_{m_1}(i), \vec{x}_{m_2}(i), P^{(i)}) \stackrel{Lem.~\ref{lem:opt_rev}}{\leq} \sum_i \rate_{P^{(i)}}.
    \end{equation}
    To prove the inequality in the other direction, we simply reverse the entire argument. For each channel $P^{(i)}$, choose $(x_1^{(i)}, \ldots, x_M^{(i)})$ to attain equality in \eqref{eq:opt_rev}. Defining $\vec{x}_m = (x_m^{(i)})_{i=1}^k$ then gives
    \begin{equation}
    \rate_{\hat{P}} \geq \frac{2}{M(M-1)} \sum_{1\leq m_1 \leq m_2 \leq M} \db(\vec{x}_{m_1}, \vec{x}_{m_2}, \hat{P}) = \sum_i \rate_{P^{(i)}},
    \end{equation}
    where both steps are based on \eqref{eq:opt_rev} holding with equality (first for $\hat{P}$, and then for the individual $\rate_{P^{(i)}}$).
\end{proof}

\subsection{Feedback Error Exponent}

In the well-studied problem of 1-hop communication with feedback, the sender is given access to all of the previous channel outputs before sending the next channel input.  For a discrete memoryless channel $P$, we define the minimum probability of error with feedback $\pe^f(n, M, P)$, and define the error exponent of the channel with $M$ messages via
\begin{equation}
\feedbackrate_{M,P} = \liminf_{T\rightarrow \infty} -\frac{1}{n} \pe^f(n,M,P).
\end{equation}

Clearly, $\feedbackrate_{M,P} \geq \rate_{M,P}$. In general, feedback error exponents may be strictly higher than non-feedback ones, even for the binary symmetric channel \cite{berlekamp_feedback}. However, the two error exponents match for arbitrary DMCs when $M=2$:
\begin{theorem}
    {\em \cite{berlekamp_feedback, baris_feedback}} For any DMC $P$, feedback does not improve the error exponent when $M=2$, i.e.
    \begin{equation}
    \feedbackrate_{2,P} = \rate_{2,P}.
    \end{equation}
    \label{thm:feedback_m_2}
\end{theorem}

\subsection{Properties of Chernoff Divergence}

The following lemma concerns the Chernoff divergence of a composite channel:
\begin{lemma}
    {\em \cite[Lemma 7]{multibit}} Let $P_1, P_2$ be channels such that $\calY_{P_1} \subseteq \calX_{P_2}$. For all $x, x' \in \calX_{P_1}$, we have
    \begin{equation}
        \dc(x, x', P_2 \circ P_1, s) \geq \min_{y, y' \in \calY_{P_1}} \big\{ \dc(y, y', P_2, s) -(1-s) \log P_1(y|x) - s \log P_1(y'|x') \big\} - 2 \log |\calY_{P_1}|.
    \end{equation}
    \label{lem:distributive}
\end{lemma}

The next result relates the error probability of a likelihood ratio test with the Chernoff divergence.  We expect that this result is standard, but we provide a short proof for completeness.

\begin{lemma}
	Let $P$ and $Q$ be two distributions over the same alphabet $\calX$. For each $x\in\calX$, define the likelihood ratio by $L(x) = \frac{P(x)}{Q(x)}$. Then
	\begin{equation}
		-\log \p_{x \sim P} (L(x) \leq L_0) \geq \db(P, Q) - \frac12 \log L_0.
	\end{equation}
	\label{lem:lrt}
\end{lemma}
\begin{proof}
	Consider the set
	\begin{equation}
		S = \{ x \,|\, L(x) \leq L_0\}.
	\end{equation}
	For all $x \in S$, $P(x) \leq L_0 \cdot Q(x)$, and therefore
	\begin{equation}
		P(x) = \sqrt{P(x)} \cdot \sqrt{P(x)} \leq \sqrt{P(x)} \cdot \sqrt{L_0 \cdot Q(x)},
	\end{equation}
	which gives
	\begin{equation}
		\p_{x \sim P} (L(x) \leq L_0) = \sum_{x \in S} P(x) \leq \sum_{x \in S} \sqrt{L_0 \cdot P(x) \cdot Q(x)} \leq \exp(-\db(P, Q)) \cdot \sqrt{L_0}.
	\end{equation}
\end{proof}

\section{The Series Case for Theorem \ref{thm:main}}
\label{sec:series}

In this section, we handle the special case in which the nodes are arranged in a series (i.e., a line graph).  That is, the nodes are labeled $\{0,1,\ldots, |E|\}$, and the edges are $\{(j, j+1): 0\leq j < |E|\}$. We refer to the corresponding channels as $P^{(1)}, P^{(2)}, \ldots P^{(|E|)}$, with corresponding 1-hop error exponents $\rate^{(1)}, \rate^{(2)},\ldots \rate^{(|E|)}$. All channels are assumed to be pairwise reversible in this section.

This series setup directly generalizes previous work on the 2-hop setting in \cite{onebit,jog2020teaching,teachlearn,multibit}, particularly our work on the the multi-bit setting \cite{multibit}.  Specifically, Theorem \ref{thm:main} for the series setup generalizes the main result therein by extending it from 2 hops to any constant number of hops.  The protocol that we adopt has some significant differences to \cite{multibit}, which we discuss in Section \ref{sec:prev}.  On the other hand, a key similarity is that we still adopt a block-structured strategy.

\subsection{Sequential Block Transmission}

We momentarily turn our attention to protocols in which nodes transmit blocks to each other \emph{one after the other} without any parallel transmission.  This idea would be highly wasteful if used in a standalone manner, but the idea (in the next subsection) will be to chain many instances of this approach together in parallel so that a negligible fraction of the total block length $n$ is ultimately wasted.

For a series graph, define a {\em sequential block transmission} with block size $B$ as follows:
\begin{itemize}
    \item For each edge $e=(u,v)$, $u$ transmits to $v$ using $B$ channels across $e$;
    \item This process occurs sequentially; transmission over an earlier edge (i.e., closer to the source) must finish before transmission over the next edge commences.
\end{itemize}
In Section \ref{sec:series_proof}, we will prove the following key lemma for this setting:
\begin{lemma}
    Given a series graph, consider a sequential block transmission with block size $B$. Let $\vec{y}$ be the length-$B$ sequence received by the final node $|E|$, and let $|f|$ be the maximum flow (since we are in the series setup, this is equal to $\min_j \rate^{(j)}$). Then there exists a protocol such that for all $m_1 \neq m_2$,
    \begin{equation}
    \db\big((\vec{y}|m=m_1), (\vec{y}|m=m_2)\big) \geq B(|f| - o_B(1)),
    \end{equation}
    where $o_B(1)$ is a term that approaches zero as $B$ increases while $G$ and $M$ remain constant.
    \label{lem:series}
\end{lemma}

\subsection{Implication for our Main Problem Setting}

Returning to our main problem setting in which nodes are not required to transmit one-after-the-other, we obtain the following corollary of Lemma \ref{lem:series}.

\begin{figure}
    \centering
    \scalebox{0.85}{
        \begin{tikzpicture}
            [line cap=round,line join=round,>=triangle 45,x=1.0cm,y=1.0cm, scale=0.6]
\draw [line width=1.pt] (0,0)-- (0,1);
\draw [line width=1.pt] (0,0)-- (20,0);
\draw [line width=1.pt] (0,1)-- (20,1);
\draw [line width=1.pt] (0,3)-- (0,4);
\draw [line width=1.pt] (0,3)-- (20,3);
\draw [line width=1.pt] (0,4)-- (20,4);
\draw [line width=1.pt] (0,6)-- (0,7);
\draw [line width=1.pt] (0,6)-- (20,6);
\draw [line width=1.pt] (0,7)-- (20,7);
\draw [line width=1.pt] (4,0)-- (4,1);
\draw [line width=1.pt] (0,0)-- (20,0);
\draw [line width=1.pt] (0,1)-- (20,1);
\draw [line width=1.pt] (4,3)-- (4,4);
\draw [line width=1.pt] (0,3)-- (20,3);
\draw [line width=1.pt] (0,4)-- (20,4);
\draw [line width=1.pt] (4,6)-- (4,7);
\draw [line width=1.pt] (0,6)-- (20,6);
\draw [line width=1.pt] (0,7)-- (20,7);
\draw [line width=1.pt] (8,0)-- (8,1);
\draw [line width=1.pt] (0,0)-- (20,0);
\draw [line width=1.pt] (0,1)-- (20,1);
\draw [line width=1.pt] (8,3)-- (8,4);
\draw [line width=1.pt] (0,3)-- (20,3);
\draw [line width=1.pt] (0,4)-- (20,4);
\draw [line width=1.pt] (8,6)-- (8,7);
\draw [line width=1.pt] (0,6)-- (20,6);
\draw [line width=1.pt] (0,7)-- (20,7);
\draw [line width=1.pt] (12,0)-- (12,1);
\draw [line width=1.pt] (0,0)-- (20,0);
\draw [line width=1.pt] (0,1)-- (20,1);
\draw [line width=1.pt] (12,3)-- (12,4);
\draw [line width=1.pt] (0,3)-- (20,3);
\draw [line width=1.pt] (0,4)-- (20,4);
\draw [line width=1.pt] (12,6)-- (12,7);
\draw [line width=1.pt] (0,6)-- (20,6);
\draw [line width=1.pt] (0,7)-- (20,7);
\draw [line width=1.pt] (16,0)-- (16,1);
\draw [line width=1.pt] (0,0)-- (20,0);
\draw [line width=1.pt] (0,1)-- (20,1);
\draw [line width=1.pt] (16,3)-- (16,4);
\draw [line width=1.pt] (0,3)-- (20,3);
\draw [line width=1.pt] (0,4)-- (20,4);
\draw [line width=1.pt] (16,6)-- (16,7);
\draw [line width=1.pt] (0,6)-- (20,6);
\draw [line width=1.pt] (0,7)-- (20,7);
\draw [line width=1.pt] (20,0)-- (20,1);
\draw [line width=1.pt] (0,0)-- (20,0);
\draw [line width=1.pt] (0,1)-- (20,1);
\draw [line width=1.pt] (20,3)-- (20,4);
\draw [line width=1.pt] (0,3)-- (20,3);
\draw [line width=1.pt] (0,4)-- (20,4);
\draw [line width=1.pt] (20,6)-- (20,7);
\draw [line width=1.pt] (0,6)-- (20,6);
\draw [line width=1.pt] (0,7)-- (20,7);
\draw [->,line width=1.pt] (0,8) -- (21,8);
\draw (22,8) node{time};
\draw (0, 7.5) node{1};
\draw (3.8, 7.5) node{$B$};
\draw (7.7, 7.5) node{$2B$};
\draw (11.7, 7.5) node{$3B$};
\draw (15.7, 7.5) node{$4B$};
\draw (19.7, 7.5) node{$5B$};
\draw (0, -0.5) node{1};
\draw (3.8, -0.5) node{$B$};
\draw (7.7, -0.5) node{$2B$};
\draw (11.7, -0.5) node{$3B$};
\draw (15.7, -0.5) node{$4B$};
\draw (19.7, -0.5) node{$5B$};
\draw [->,line width=1.pt] (2,6) -- (6,4);
\draw [->,line width=1.pt] (14,3) -- (18,1);
\draw [->,line width=1.pt] (6,6) -- (10,4);
\draw [->,line width=1.pt] (6,3) -- (10,1);
\draw [->,line width=1.pt] (10,6) -- (14,4);
\draw [->,line width=1.pt] (10,3) -- (14,1);
\draw [fill=blue] (0,6) rectangle (4,7);
\draw [fill=blue] (4,3) rectangle (8,4);
\draw [fill=blue] (8,0) rectangle (12,1);
\draw [fill=red] (4,6) rectangle (8,7);
\draw [fill=red] (8,3) rectangle (12,4);
\draw [fill=red] (12,0) rectangle (16,1);
\draw [fill=green] (8,6) rectangle (12,7);
\draw [fill=green] (12,3) rectangle (16,4);
\draw [fill=green] (16,0) rectangle (20,1);
\draw (21.97,6.5) node{node 1 (source)};
\draw (20.9,3.5) node{node 2};
\draw (22.5,0.5) node{node 3 (destination)};
        \end{tikzpicture}
    }
    \caption{An illustration of how a total of $n$ time steps can be broken down into $n/B-|E|$ sequential block transmissions. In this diagram, $n=5B$ and $|E|=2$.}
    \label{fig:block-protocol}
\end{figure}
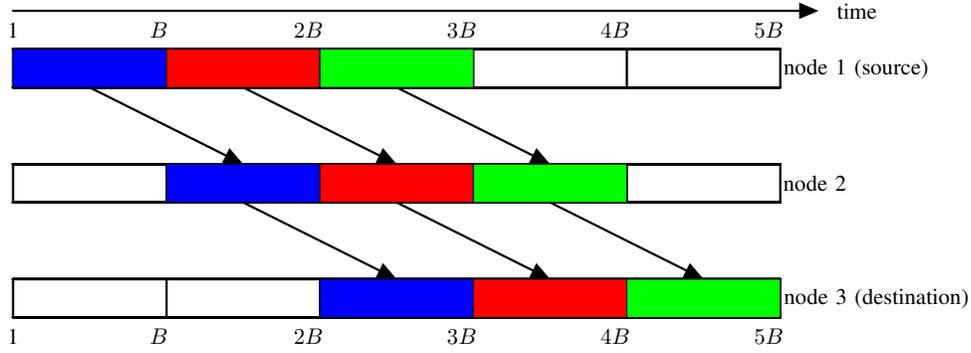
    
\begin{cor}
    Given a series graph with $n$ time steps, there exists a protocol such that for any $m_1\neq m_2$, the length-$n$ sequence $\vec{y}_{\rm all}$ received at node $|E|$ satisfies
    \begin{equation}
    \db\big((\vec{y}_{\rm all}|m=m_1), (\vec{y}_{\rm all}|m=m_2)\big) \geq (n-|E|B) \cdot (|f|-o_B(1)),
    \end{equation}
    where $o_B(1)$ is a term that approaches zero as $B$ increases while $G$ and $M$ remain constant (without any dependence on $n$).
    \label{cor:series}
\end{cor}
\begin{proof}
    Consider the process of chaining several sequential block transmissions in parallel, as depicted in Figure \ref{fig:block-protocol}.  If there are $t$ sequential block transmissions of length $B$, then they can clearly be executed in $(t+|E|)\cdot B$ time steps.  In other words, with a block length of $n$, we can perform $n/B-|E|$ sequential block transmissions (up to asymptotically negligible rounding issues).

    For each such sequential block transmission, we use the protocol that achieves Lemma \ref{lem:series}. Since memory is not stored between the blocks, the blocks are conditionally independent given the source message. Let $\vec{y}(t)$ represent the $t$-th block received by the final node, and let $\vec{y}_{\rm all} = (\vec{y}(1), \vec{y}(2), \ldots, \vec{y}(n/B-|E|))$.  Then, using tensorization (Lemma \ref{lem:iid}) and Lemma \ref{lem:series}, we have for all $m_1 \neq m_2$ that
    \begin{eqnarray}
    & &\db\big((\vec{y}_{\rm all}|m=m_1), (\vec{y}_{\rm all}|m=m_2)\big) \nonumber\\
    & = & \sum_{t=1}^{n/B-|E|} \db\big((\vec{y}(t)|m=m_1), (\vec{y}(t)|m=m_2)\big) \\
    & \geq & (n/B-|E|) B(|f|-o_B(1)) \\
    &=& (n-|E|B)\cdot(|f|-o_B(1)),
    \end{eqnarray}
    as required.
\end{proof}

Then, the error exponent is bounded by
\begin{eqnarray}
& &\liminf_{n \to \infty} -\frac{1}{n} \log \pe \\
& \geq &\liminf_{n \to \infty} \min_{m_1 \neq m_2} \db\big((\vec{y}_{\rm all} | m = m_1), (\vec{y}_{\rm all} | m = m_2)\big) \label{eq:bhat_bound} \\
& = & \liminf_{n \to \infty} \frac{n-|E|B}{n} (|f| - o_B(1)) \label{eq:apply_cor} \\
& = & |f|-o_B(1),
\end{eqnarray}
where \eqref{eq:bhat_bound} follows from Lemma \ref{lem:chernoff}, and \eqref{eq:apply_cor} follow from Corollary \ref{cor:series}.  This proves that the error exponents arbitrarily close to $|f|$ are achievable, which proves Theorem \ref{thm:main} for the series case.

\subsection{Proof of Lemma \ref{lem:series}} \label{sec:series_proof}

We first show that it suffices to prove Lemma \ref{lem:series} under an additional assumption that will be convenient to work with:

\begin{lemma}
    Suppose that Lemma \ref{lem:series} holds under the additional restriction that each channel $P^{(j)}$ has $M$ inputs $\{1,2,\ldots, M\}$ such that $\db(m_1, m_2, P^{(j)}) \geq |f|$ for all $m_1\neq m_2$.  Then the general form of Lemma \ref{lem:series} as stated also holds.
    \label{lem:separated_inputs}
\end{lemma}

\begin{proof}
Consider the general setup of Lemma \ref{lem:series} with arbitrary $P^{(j)}$. 
 By Lemma \ref{lem:db_codebook}, for each $j$, there exist codewords  $\vec{x}^{(j)}_1, \vec{x}^{(j)}_2, \ldots, \vec{x}^{(j)}_M$ of length $\ell = M!$ such that 
\begin{equation}
    \db(\vec{x}^{(j)}_{m_1}, \vec{x}^{(j)}_{m_2}, (P^{(j)})^\ell) \geq \ell \cdot \rate_{M,P^{(j)}}. \label{eq:db_ell}
\end{equation}
Let $Q^{(j)}$ be the restriction of $(P^{(j)})^\ell$ to the codewords  $\vec{x}^{(j)}_1, \vec{x}^{(j)}_2, \ldots, \vec{x}^{(j)}_M$.  We first show that the channels $Q^{(j)}$ satisfy the extra requirement stated in Lemma \ref{lem:separated_inputs}.

By Lemma \ref{lem:block_reduction}, $Q^{(j)}$ is pairwise reversible. Moreover, for each $m_1 \neq m_2$, we have from \eqref{eq:db_ell} that
\begin{equation}
\db(\vec{x}^{(j)}_{m_1}, \vec{x}^{(j)}_{m_2}, Q^{(j)}) = \db(\vec{x}^{(j)}_{m_1}, \vec{x}^{(j)}_{m_2}, (P^{(j)})^\ell) \geq \ell \cdot \rate_{M,P^{(j)}}. \label{eq:db_Q}
\end{equation}
Then, we deduce from Lemma \ref{lem:opt_rev} (with the $M$ inputs of $Q^{(j)}$ being substituted for $(x_1,x_2,\ldots, x_M)$) that $\rate_{M,Q^{(j)}} \geq \ell \cdot \rate_{M,P^{(j)}}$.  Since the maxflow is $|f| = \min_{j} \rate_{M,P^{(j)}}$ in the series setup, this gives $\rate_{M,Q^{(j)}} \ge \ell \cdot |f|$.
Then, let $|f'|$ be the maximum flow for the sequence of channels $Q^{(1)}, \ldots, Q^{(|E|)}$. Since $\rate_{M,Q^{(j)}} \geq \ell \cdot |f|$ for all $j$, the new maximum flow satisfies $|f'| \geq \ell \cdot |f|$. On the other hand, by Lemma \ref{lem:block_reduction}, we have $\rate_{Q^{(j)}} \leq \ell \cdot \rate_{P^{(j)}}$, which gives
\begin{equation}
|f'| = \min_j \rate_{Q^{(j)}} \leq  \min_j \ell \cdot \rate_{P^{(j)}} \leq \ell \cdot |f|.
\end{equation}
Combining these upper and lower bounds on $|f'|$, we conclude that $|f'| = \ell \cdot |f|$, and from \eqref{eq:db_Q} we deduce that the channels $Q^{(j)}$ satisfy the additional requirement imposed in the statement of Lemma \ref{lem:separated_inputs} (with $|f|$ renamed to $|f'|$).

Now suppose that Lemma \ref{lem:series} holds for the new sequence of channels $Q^{(1)}, \ldots, Q^{(|E|)}$, and let $\vec{y}'$ denote the corresponding length-$B$ sequence received by the destination node. This means that there exists a sequential block transmission protocol of block size $B$ such that for any $m_1 \neq m_2$,
\begin{equation}
    \db\big((\vec{y}'|m=m_1), (\vec{y}'|m=m_2)\big) \geq B(|f'| - o_B(1)) \geq B(\ell |f|-o_B(1)).
\end{equation}

Any sequential block transmission with block size $B$ using the new channels $Q^{(1)}, \ldots, Q^{(|E|)}$ can easily be converted to a sequential block transmission with block size $B \cdot \ell$ over the original channels $P^{(1)}, \ldots, P^{(|E|)}$ by simply expanding instances of the new letter $m$ by the codeword $\vec{x}_M$. Therefore, there exists a sequential block transmission of block size $B \cdot \ell$ such that 
 for any $m_1 \neq m_2$,
\begin{equation}
    \db\big((\vec{y}|m=m_1), (\vec{y}|m=m_2)\big) \geq B(\ell|f| - o_B(1)).
\end{equation}
Since $\ell$ is a constant, we can appropriately rescale such that $B \ell$ above is renamed to $B$ (the overall sequential block transmission length), and we conclude that Lemma \ref{lem:series} holds for the original channels $P^{(1)}, \ldots, P^{(|E|)}$.
\end{proof}
For the rest of this section, we work under the extra assumption stated in Lemma \ref{lem:separated_inputs}. We also assume that $B$ is even to avoid rounding issues.

\subsubsection{Description of the Forwarding Protocol}

We consider the following protocol for a single sequential block transmission of length $B$:
\begin{itemize}
    \item Every node $j$ ($0 \leq j \leq |E|$) is assigned a state consisting of two integers $(m^{(j)},\ell^{(j)})$ with $1\leq m^{(j)} \leq M$ and $0\leq \ell^{(j)} \leq B/2$, which we will refer to as $(m,\ell)$ when there is no ambiguity.  
    A state of $(m^{(j)},\ell^{(j)})$ roughly indicates that the node believes $m^{(j)}$ to be the most likely value of $m$, where larger values of $\ell^{(j)}$ indicate higher confidence, whereas a smaller value of $\ell^{(j)}$ indicates that there is another competing value $m' \neq m^{(j)}$ that is also relatively likely (we make this more precise below).
    \item Node 0 is assigned $m^{(0)}=m$ and $\ell^{(0)} = B/2$, where $m$ is the actual message.
    \item If node $j$ is assigned a state $(m,\ell)$, then this node sends out a $B/2 + \ell$ copies of $m$ followed by $B/2 - \ell$ copies of $m+1$ (unless $m=M$, in which case we wrap around and replace $m+1$ by $1$ here).  We refer to the corresponding string as $\vec{x}_{m,\ell}$, or $\vec{x}^{(j)}_{m,\ell}$ when we want to emphasize which node is being considered.
    \item Node $j$ ($j\geq 1$), upon receiving $\vec{y}^{(j)}$ from node $j-1$, computes the likelihoods $\p(\vec{y}^{(j)}|m^{(0)}=m)$ for each possible message $1\leq m \leq M$.  Let $m^*, m^{**}$ be the resulting most likely and second most likely value of $m$ respectively (with arbitrary tie-breaking).
    \item Node $j$ is now assigned a state of
    \begin{equation}
    m^{(j)} = m_1, \ell^{(j)} = \min\left(\frac B2, \Bigg\lfloor \frac{1}{4|f|} \cdot \log \frac{\p(\vec{y}^{(j)}|m^*)}{\p(\vec{y}^{(j)}|m^{**})}\Bigg\rfloor\right).
    \label{eq:state_assign}
    \end{equation}
    These equations are chosen such that $m^{(j)}$ reflects the most likely source message, and $\ell^{(j)}$ appropriately reflects the confidence level such that the most likely message $m^{(j)}$ is roughly $\exp(4\ell^{(j)}|f|)$ times more likely than the next value of $m$.
\end{itemize}

\subsubsection{Analysis of the Strategy}
We start by defining a pseudometric on the states by
\begin{equation}
d((m_1,\ell_1), (m_2, \ell_2)) = 
\begin{cases}
    |f| |\ell_1 - \ell_2| & m_1 = m_2 \\
    |f| |\ell_1 + \ell_2| & m_1 \neq m_2.
\end{cases}
\label{eq:metric}
\end{equation}
To see that $d$ defined above is a pseudometric, note that $d$ is equivalent to the $\ell_1$-norm after applying the mapping $(m,\ell) \rightarrow (0,\ldots, 0, \ell |f|, 0,\ldots, 0)$ (with $\ell |f|$ occurring at position $m$).  The reason for being a pseudometric rather than a metric is that $d((m_1, 0), (m_2,0))=0$.

Intuitively, $d$ is a measure of how far apart the beliefs $(m_1, \ell_1)$ and $(m_2, \ell_2)$ are. If $m_1 = m_2$, then these two beliefs are close when $\ell_1$ and $\ell_2$ are close. If $m_1 \neq m_2$, the beliefs are close only when $m_1$ and $m_2$ are both comparatively likely, which corresponds to the case where both $\ell_1$ and $\ell_2$ are small.

We use $d$ to bound the Bhattacharyya distance associated with two codewords:
\begin{eqnarray}
\db(\vec{x}^{(j)}_{m_1,\ell_1}, \vec{x}^{(j)}_{m_2,\ell_2}, (P^{(j)})^B) & = & \sum_{i=1}^B \db(\vec{x}^{(j)}_{m_1,\ell_1}(i), \vec{x}^{(j)}_{m_2,\ell_2}(i), P^{(j)}) \label{eq:tenorize_expand}\\
& \geq & |f| d_H(\vec{x}^{(j)}_{m_1,\ell_1}(i), \vec{x}^{(j)}_{m_2,\ell_2}(i))\label{eq:hamming_bound} \\
&\geq & d((m_1,\ell_1), (m_2, \ell_2)),
\label{eq:metric_bound}
\end{eqnarray}
where:
\begin{itemize}
    \item \eqref{eq:tenorize_expand} follows from tensorization (Lemma \ref{lem:iid}).
    \item \eqref{eq:hamming_bound} comes from the fact that the number of indices in which $\vec{x}^{(j)}_{m_1,\ell_1}(i)$ and $\vec{x}^{(j)}_{m_2,\ell_2}(i)$ differ is simply their Hamming distance, and each of these corresponding terms are at least $|f|$ due to the extra assumption in Lemma \ref{lem:separated_inputs}.
    \item To see why \eqref{eq:metric_bound} holds, we use the definition of $\vec{x}^{(j)}_{m,\ell}$ and handle the cases $m_1 = m_2$ and $m_1 \neq m_2$ separately. If $m_1 = m_2$, then
    \begin{equation}
        d_H(\vec{x}^{(j)}_{m_1,\ell_1},\vec{x}^{(j)}_{m_2,\ell_2}) = |\ell_1 - \ell_2|,
    \end{equation}
    which matches \eqref{eq:metric}.  
    If $m_1 \neq m_2$, then the first $B/2 + \min(\ell_1, \ell_2)$ symbols of the two codewords consist of only $m_1$ and only $m_2$ respectively, so that
    \begin{equation}
        d_H(\vec{x}^{(j)}_{m_1,\ell_1},\vec{x}^{(j)}_{m_2,\ell_2})  \geq B/2 + \min(\ell_1, \ell_2) \geq \max(\ell_1, \ell_2) + \min(\ell_1, \ell_2) = \ell_1 + \ell_2,
    \end{equation}
    which again matches \eqref{eq:metric}.
\end{itemize}
Observe that the sequence of states taken by the nodes forms a Markov process.  We establish a bound on the probability of achieving any intermediate state given the beginning state; intuitively, the following result states that given the original state $(m^{(0)},\ell^{(0)})=(m,B/2)$, the probability of seeing another state $(m^{(j)}, \ell^{(j)})$ is exponentially small in the distance between the two states defined by \eqref{eq:metric}.

\begin{lemma}
For all $j$, all $m_1 \neq m_2$, and all $(m',\ell')$, we have
\begin{equation}
-\log \p\big(m^{(j)}=m',\ell^{(j)}=\ell'|m=m_1\big) \geq
2d\big((m_1, B/2),(m',\ell')\big) - 2j\log (M(B+1)) - 2j|f| - j\log(M-1)
\label{eq:transition}
\end{equation}
and
\begin{equation}
\db\big(\big(\vec{y}^{(j+1)}|m=m_1\big), \big(\vec{y}^{(j+1)}|m=m_2\big)\big) \geq B|f| - 2(j+1)\log (M(B+1)) - 2j|f| - j\log(M-1).
\label{eq:chernoff_transition}
\end{equation}
\label{lem:transition}
\end{lemma}
\begin{proof}
We use an induction argument, with the base case being that \eqref{eq:transition} trivially holds when $j=0$: Either the left-hand side is infinite due to a probability of zero, or the left-hand side is zero due to a probability of one (i.e., $(m',\ell') = (m_1,B/2)$), in which case the first term on the right-hand side is zero.

We proceed by induction on $j$, and break the inductive argument into the following two claims:
\begin{itemize}
\item If \eqref{eq:transition} holds for $j$, then \eqref{eq:chernoff_transition} holds for $j$
\item If \eqref{eq:chernoff_transition} holds for $j$, then \eqref{eq:transition} holds for $j+1$.
\end{itemize}
Start by assuming that \eqref{eq:transition} holds for some $j$. Observe that $m \rightarrow (m^{(j)}, \ell^{(j)}) \rightarrow \vec{y}^{(j+1)}$ forms a Markov chain. Let $P_1$ be the channel $m \rightarrow (m^{(j)}, \ell^{(j)})$ and $P_2$ be the channel $(m^{(j)}, \ell^{(j)}) \rightarrow \vec{y}^{(j+1)}$; applying Lemma \ref{lem:distributive} with $s=1/2$ then gives for any $m_1 \ne m_2$ that
\begin{eqnarray}
& & \db\left(\big(\vec{y}^{(j+1)}|m^{(0)}=m_1\big), \big(\vec{y}^{(j+1)}|m^{(0)}=m_2\big)\right) \nonumber \\ 
&\geq &\min_{m',m'',\ell',\ell''} \Big\{\db(\vec{x}_{m',\ell'}, \vec{x}_{m'',\ell''}, (P^{(j)})^B) - \frac12 \log \p(m',\ell' | m^{(0)}=m_1) - \frac12 \log \p(m'',\ell'' | m^{(0)}=m_2) \Big\} \nonumber\\
& &- 2\log(M(B+1)).
\label{eq:expand_db}
\end{eqnarray}
Since \eqref{eq:transition} holds for $j$ (and for all $(m',\ell')$ and $m_1$), we can combine it with \eqref{eq:metric_bound} to simplify the expression inside the minimization clause:
\begin{eqnarray}
& &\db(\vec{x}_{m',\ell'}, \vec{x}_{m'',\ell''}, (P^{(j)})^B) - \frac12 \log \p(m',\ell' | m=m_1) - \frac12 \p(m'',\ell'' | m=m_2) \nonumber \\
& \geq & d\big((m',\ell'), (m'',\ell'')\big) + d\big((m_1, B/2), (m',\ell')\big) + d\big((m_2, B/2), (m',\ell')\big) \\
& &- 2j\log (M(B+1)) - 2j|f|- j\log(M-1) \nonumber \\
&\geq & d\big((m_1,B/2), (m_2,B/2)\big) - 2j\log (M(B+1)) - 2j|f|- j\log(M-1) \label{eq:last_ineq} \\ 
&=&  B|f| - 2j\log (M(B+1)) - 2j|f| - j\log(M-1),
\label{eq:simplify_minimization}
\end{eqnarray}
where \eqref{eq:last_ineq} follows from non-negativity (first $d$ term) and the triangle inequality (second and third $d$ terms) for the pseudometric $d$, and the last step uses the definition of $d$ with $m_1 \ne m_2$. By substituting \eqref{eq:simplify_minimization} into \eqref{eq:expand_db}, we conclude that \eqref{eq:chernoff_transition} holds for $j$.

We now turn to the second dot point mentioned at the start of the proof. We will split this into two cases.  We note that in \eqref{eq:transition} we are interested in the event that $(m^{(j)},\ell^{(j)}) = (m',\ell')$, and for this to happen, $m'$ must be the value assigned to $m^*$ used in \eqref{eq:state_assign}, i.e., $m'$ must be the most likely message given $\vec{y}^{(j)}$.  We will use this fact in both cases.

\textbf{Case 1 ($m'=m_1$).} If $\ell' = B/2$, then $d((m_1,B/2), (m', \ell'))=0$, so that the RHS of \eqref{eq:transition} is upper bounded by 0, and therefore \eqref{eq:transition} trivially holds. Therefore, we assume that $\ell' < B/2$.

For each $\vec{y}^{(j+1)}$ and each $m'' \neq m_1$, define
\begin{equation}
L_{m''}(\vec{y}^{(j+1)}) = \frac{\p(\vec{y}^{(j+1)}|m^{(0)}=m_1)}{\p(\vec{y}^{(j+1)}|m^{(0)}=m'')}.
\end{equation}
By \eqref{eq:state_assign} and the assumption $\ell^{(j+1)} < B/2$, there exists $m''$ (which can be chosen to be the second most likely value of $m^{(0)}$ given $\ell^{(j+1)}$, noting that the most likely value is $m'=m_1$) such that
\begin{equation}
\frac{1}{4|f|} \log L_{m''}(\vec{y}^{(j+1)}) < \ell^{(j+1)}+1.
\end{equation}
Combining this with a union bound over all $m''$, we have
\begin{eqnarray}
& &\p(m^{(j+1)}=m',\ell^{(j+1)}=\ell'|m^{(0)}=m_1) \\
&\leq & (M-1) \max_{m'' \neq m_1} \p\left(L_{m''}(\vec{y}^{(j+1)}) < \exp(4|f|(\ell'+1)) \,\Big|\, m^{(0)}=m_1\right).
\end{eqnarray}
We can now conclude that
\begin{eqnarray}
& & -\log\p(m^{(j+1)}=m',\ell^{(j+1)}=\ell'|m^{(0)}=m_1) \nonumber \\
&\geq & -\log\bigg(\max_{m'' \neq m_1}\p\left(L_{m''}(\vec{y}^{(j+1)}) < \exp(4|f|(\ell'+1)) \,\Big|\, m^{(0)}=m_1\right)\bigg) -\log(M-1) \label{eq:case1_step1}\\
&\geq  & \min_{m'' \neq m_1} \bigg(\db\left((\vec{y}^{(j+1)}|m^{(0)}=m_1), (\vec{y}^{(j+1)}|m^{(0)}=m'')\right) \bigg) - 2|f|(\ell'+1)  - \log(M-1)\label{eq:use_lrt} \\
&\geq & B|f| - 2(j+1)\log (M(B+1)) - 2j|f| - j\log(M-1) - 2|f|(\ell'+1) -\log(M-1) \label{eq:use_chernoff} \\
& = & (B-2\ell')|f|-2(j+1)\log(M(B+1))-2(j+1)|f| - (j+1)\log(M-1)\\
& = & 2d\big((m_1, B/2), (m', \ell')\big)-2(j+1)\log(M(B+1))-2(j+1)|f| - (j+1)\log(M-1). \label{eq:use_metric}
\end{eqnarray}
where \eqref{eq:use_lrt} follows from Lemma \ref{lem:lrt}, \eqref{eq:use_chernoff} follows from the inductive assumption \eqref{eq:chernoff_transition}, and \eqref{eq:use_metric} follows from the definition of $d$ with $m' = m_1$.  Thus, \eqref{eq:transition} holds for $j+1$ as desired.

\textbf{Case 2 ($m'\neq m_1$).} For each $\vec{y}^{(j+1)}$, define
\begin{equation}
L(\vec{y}^{(j+1)}) = \frac{\p(\vec{y}^{(j+1)}|m^{(0)}=m_1)}{\p(\vec{y}^{(j+1)}|m^{(0)}=m')}.
\end{equation}
Let $m''$ (which may depend on $\vec{y}^{(j+1)}$) be the second most likely value of $m^{(0)}$ (noting that $m'$ is the most likely value).  We then have
\begin{equation}
\frac{1}{4|f|} \log L(\vec{y}^{(j+1)}) \leq \frac{1}{4|f|} \log \frac{\p(\vec{y}^{(j+1)}|m^{(0)}=m'')}{\p(\vec{y}^{(j+1)}|m^{(0)}=m')} \leq -\ell^{(j+1)},
\end{equation}
where the first inequality uses the fact that the likelihood for $m_1$ is at most that of $m''$, and the second inequality follows from \eqref{eq:state_assign} after lowing bounding the minimum therein by the second term.

We now proceed similarly to Case 1:
\begin{eqnarray}
& &-\log \p(m^{(j+1)}=m',\ell^{(j+1)}=\ell'|m^{(0)}=m_1) \\
&\geq &-\log \p\left(L(\vec{y}^{(j+1)}) \leq \exp(-4|f|\ell') \,\Big|\, m^{(0)}=m_1\right)\nonumber \\
&\geq & \db\left((\vec{y}^{(j+1)}|m^{(0)}=m'), (\vec{y}^{(j+1)}|m^{(0)}=m_1)\right) + 2|f|\ell'\nonumber \\
&\geq & B|f| - 2(j+1)\log (M(B+1)) - 2j|f| - j\log(M-1) + 2|f|\ell'\nonumber \\
& = & (B+2\ell')|f| - 2(j+1)\log (M(B+1)) - 2j|f| -j\log(M-1)\nonumber \\
& \geq & 2d\big((m',\ell'),(m_1, B/2)\big)-2(j+1)\log(M(B+1))-2(j+1)|f| - (j+1)\log(M-1),
\end{eqnarray}
where each step follow similarly to its counterpart in \eqref{eq:case1_step1}--\eqref{eq:use_metric}. 
We have thus established that \eqref{eq:transition} holds for $j+1$, which completes the proof of Lemma \ref{lem:transition}.
\end{proof}

The key result in Lemma \ref{lem:series} now immediately follows from \eqref{eq:chernoff_transition} applied to $j=|E|-1$, and by recalling that the lemma is stated with respect to growing $B$ with fixed $|E|$ and $M$.

\subsection{Comparison to Previous Work} \label{sec:prev}

In our previous paper \cite{multibit}, we proved the 2-hop version of Lemma \ref{lem:series}. In \cite{multibit}, the problem is reduced to the case $|\calX_P| = M$ (similarly to Lemma \ref{lem:separated_inputs}) and then handled as follows:
\begin{itemize}
    \item The relay node calculates the KL-divergence between the empirical distribution of the $B$ received symbols against the output distribution associated with each message.
    \item The relay then sends a sorted sequence with more occurrences of the symbol whose associated KL-divergence is smallest, and an equal number of occurrences of the rest.
\end{itemize}
Our approach in the present paper is instead based on a likelihood ratio (see \eqref{eq:state_assign}), and upon calculating this, a sorted sequence is sent containing at most two distinct values.

At first glance, one may have expected that using likelihood ratios would be difficult, as there can be $\binom{M}{2}$ likelihood ratios of interest. However, our approach here shows that using the likelihood ratios between only the two most likely values is sufficient.  Overall, our approach here is more similar to the one used in \cite{onebit} (for general binary-input channels rather than the BSC), which directly computes the likelihood ratio but only handles the 2-hop setting and requires $M=2$.

\section{Proof of Theorem \ref{thm:main} (Main Achievability Bound)}
\label{sec:proof_main}

In this section, we prove Theorem \ref{thm:main} in the general case, using the series case from Section \ref{sec:series} as a stepping stone.    
We crucially make use of the flow decomposition theorem (Theorem \ref{thm:flow_decomp}) together with its associated paths $p_1, p_2,\ldots, p_k$ and flows $f_1, f_2,\ldots, f_k$. The overall idea is to send conditionally independent messages along different paths and only accumulate them at the end. 

\subsection{Edge-Disjoint Paths}
\label{sec:disjoint}

The protocol is simpler when the paths $p_1, p_2,\ldots, p_k$ are edge-disjoint, so we will consider this case before turning to the general case.

The key idea is to transmit information across each path independently and aggregate them only at the end. Although we assume that paths are edge-disjoint, they may still share nodes. However, each node is designed to act independently on every path it is involved in, so that information entering a node in one path does not affect other outgoing paths.

Let $\vec{y}(1), \ldots, \vec{y}(k)$ be the sequences received by the destination node along the $k$ different paths after $n$ time steps, and let $\vec{y}_{\rm all} = (\vec{y}(1), \ldots, \vec{y}(k))$. By Corollary \ref{cor:series}, for all $i$ and all $m_1 \neq m_2$, there exists a (single-path) protocol such that
\begin{equation}
\db\big((\vec{y}(i)|m=m_1), (\vec{y}(i)|m=m_2)\big) \geq (n-|V|B)(|f_i|-o_B(1)),
\end{equation}
where the $|V|$ term comes from the fact that a simple path cannot have more than $|V|$ edges.  
Therefore, again letting  $B \rightarrow \infty$, we have
\begin{eqnarray}
& &\db\big((\vec{y}_{\rm all} | m=m_1), (\vec{y}_{\rm all} | m=m_2)\big) \nonumber \\
& = & \sum_{i=1}^k \db\big(\vec{y}(i)|m=m_1, \vec{y}(i)|m=m_2\big) \nonumber \\
& \geq & \sum_{i=1}^k (n-|V|B)(|f_i|-o_B(1)) \\
& = & (n-|V|B)(|f|-o_B(1)).
\end{eqnarray}
Note that here $k$ is treated as a constant, because $k$ is a property of the graph $G$ so does not depend on $B$.

Applying Lemma \ref{lem:chernoff}, we deduce that the total error exponent is bounded below by
\begin{equation}
\rate_G \geq \liminf_{n \to \infty} -\frac{1}{n}(n-|V|B)(|f|-o_B(1)) = |f|-o_B(1),
\end{equation}
which shows that error exponents arbitrarily close to $|f|$ are achievable.

\subsection{The General Case}

For the general case, we will still rely on the edge decomposition as before. Our approach is to reduce to the case where paths are edge-disjoint by considering the use of parallel edges, which are allowed to exist in all of our analysis and results.  This idea has also been used previously when studying channel capacity results for general networks, e.g., see \cite[Sec.~15.3]{gamalkim}

We first make two observations:
\begin{itemize}
    \item Suppose that there is some edge $e$ such that $P_e$ is a product channel, say $P_e = Q_1 \times Q_2$. We can delete $e$ and replace it by two parallel edges $e_1$ and $e_2$, and put the channels $Q_1$ and $Q_2$ respectively and preserve the error exponent. This is because any symbol that can be sent across $Q_1 \times Q_2$ can be separately sent across the two channels in an identical manner.
    \item For any graph $G$, let $G'$ be the resultant graph formed by replacing each $P_e$ with its $\ell$-fold product $P_e^{\ell}$. Then $\rate_{G'}\leq \ell \rate_{G}$. To see this, observe that any protocol with $n$ time steps over $G'$ can be executed with $n\cdot \ell$ time steps over $G$.
\end{itemize}

These two observations allow us to do the following:
\begin{itemize}
\item Pick a large block size $B$, and replace every channel $P_e$ with the $(B+k)$-fold product channel $P_e^{B+k}$
\item For each edge $e$, let
\begin{equation}
S_e = \{i : e \in p_i\}
\end{equation}
be the indices of the paths that pass through $e$.
For each $i \in S_e$, define
\begin{equation}
B_i = \Bigg\lceil \frac{f_i}{\sum_{j \in S_e} f_j} \cdot B \Bigg\rceil.
\end{equation}
Since there are only $k$ paths, we have $|S_e| \leq k$.  Moreover, we have
\begin{equation}
    \sum_{i \in S_e} B_i = \sum_{i \in S_e}  \Bigg\lceil \frac{f_i}{\sum_{j \in S_e} f_j} \cdot B \Bigg\rceil \leq \sum_{i \in S_e}  \Bigg( \frac{f_i}{\sum_{j \in S_e} f_j} \cdot B + 1\Bigg) = B + |S_e| \leq B + k.
\end{equation}

For each edge $e$ with corresponding channel $P_e^{B+k}$, replace this by $|S_e|$ parallel edges, where the corresponding channels are $P_e^{B_i}$ for each $i \in S_e$. We will refer to each of these new edges by $e'_i$. This is possible since $\sum_{i \in S_e} B_i \leq B+k$. If $\sum_{i\in S_e} B_i < B + k$, then we simply add a dummy (unused) channel to make up for the difference. Let $G'$ denote this new graph.
\end{itemize}

Observe that after the transformation, the capacities of the corresponding parallel edges are
\begin{equation}
\Bigg\lceil \frac{f_i}{\sum_{j \in S_e} f_j} \cdot B \Bigg\rceil \cdot \rate_e \ge B \cdot f_i,
\label{eq:k_i}
\end{equation}
noting that $\sum_{j \in S_e} f_j = f_e \leq \rate_e$ by the flow constraints.

Let $p'_i$ be the path in $G'$ formed by joining up edges of the form $e'_i$ for each $e \in p_i$. Since $p_i$ is a path and $e'_i$ joins the same pair of vertices as $e$, $p'_i$ is also a path. By \eqref{eq:k_i}, each of these edges $e'_i$ has error exponent at least $B \cdot f_i$, so that the capacity on path $p'_i$ is at least $B \cdot f_i$.  
Thus, there exist edge-disjoint paths $p'_1, p'_2, \ldots, p'_k$ on $G'$ such that the capacity on each path $p'_i$ is at least $B\cdot f_i$.

We have now reduced the problem to the edge-disjoint case studied in the previous subsection, and we conclude that we can achieve error exponents arbitrarily close to $B |f|$. Since time was scaled by a factor of $B+k$, we conclude that $\rate_{G'} \leq (B+k)\rate_G$ and therefore
\begin{equation}
\rate_G \geq \frac{1}{(B+k)}\rate_{G'} = \frac{B|f|}{(B+k)}
\end{equation}
Taking $B\rightarrow\infty$, we conclude that $\rate_G \geq |f|$.

\section{Proof of Theorem \ref{thm:main_converse} (Main Converse Bound)}
\label{sec:main_converse}

The proof of our converse bound consists of using the max-flow min-cut duality theorem (Theorem \ref{thm:duality}) to construct a cut with the same value as $|f|$. We will additionally use the following lemma, where we recall that $\feedbackrate_{M,P}$ denotes the 1-hop error exponent of a channel with full feedback:

\begin{lemma}
Consider any partition $(A,B)$ of the nodes in $G$ with the source node in $A$ and the destination node in $B$. Let $P^{(1)}, P^{(2)},\ldots, P^{(k)}$ be the channels associated with edges from $A$ to $B$, and let $\hat{P} = P^{(1)}\times P^{(2)} \times \ldots \times P^{(k)}$ be the product channel. Then for all $M$,
\begin{equation}
\rate_{M,G} \leq \feedbackrate_{M,\hat{P}}.
\end{equation}

\label{lem:cut_bound}
\end{lemma}
\begin{proof}
We compress $A$ into a single node by allowing infinite information transfer between different nodes inside $A$, and do the same for $B$.  Clearly, this modification cannot decrease the overall error exponent.

At first glance, it may appear that the net result of doing this is simply allowing $A$ to send information to $B$ via $P^{(1)}, P^{(2)},\ldots, P^{(k)}$, each a total of $n$ times, which is the same as using $\hat{P}$ for $n$ times. However, we must also account for the possibility of back-edges, which amounts to having a certain degree of feedback.   
In view of this, since any protocol that can be executed under the original setup can also be executed under the new setup with $n$ uses of $\hat{P}$ {\em with full feedback}, we conclude that $\rate_{M,G} \leq \feedbackrate_{M,\hat{P}}$.
\end{proof}

Since $\rate_{M,P}^f$ is a decreasing function of $M$, $\rate_{M,P}^f \leq \rate_{2,P}^f$. Furthermore, feedback does not improve the error exponent when $M = 2$ (Theorem \ref{thm:feedback_m_2}), and thus Lemma \ref{lem:cut_bound} immediately gives the following corollary:
\begin{cor} 
Under the definition of $\hat{P}$ in Lemma \ref{lem:cut_bound}, we have
\begin{equation}
\rate_{M,G} \leq \rate_{2,\hat{P}}.
\label{cor:better_cut_bound}
\end{equation}
\end{cor}

We can now use maxflow-mincut duality (Theorem \ref{thm:duality}) and consider a cut $(A,B)$ such such that $A$ contains the source, $B$ contains the sink, and 
\begin{equation}
\sum_{i=1}^k \rate_{2,P^{(i)}} = |f|,
\label{eq:special_cut}
\end{equation}
where $P^{(1)}, P^{(2)},\ldots, P^{(k)}$ are the channels associated with edges that begin in $A$ and end in $B$. 
Letting $\hat{P} = P^{(1)} \times P^{(2)} \times \ldots \times P^{(k)}$ be the product channel, Lemma \ref{lem:parallel} gives
\begin{equation}
\rate_{2,\hat{P}} \leq \sum_{i} \rate_{2,P^{(i)}} 
\end{equation}
which implies
\begin{equation} 
\rate_{M,G} \leq \rate_{2,\hat{P}} \leq \sum_{i} \rate_{2,P^{(i)}} = |f|
\end{equation}
as required.

\section{Achievability Bounds for General Channels} \label{sec:ach_general}

In this section, we study achievability bounds in the case that the channels need not be pairwise reversible. 

\subsection{Constant Number of Messages}

When the pairwise reversible assumption is dropped, Theorem \ref{thm:main} can fail to hold even in the 2-hop setting corresponding to a graph with 3 nodes and 2 edges \cite[Thm.~3]{multibit}. However, it is possible to prove some weaker achievability bounds. Recalling $\tilde{\rate}$ from Definition \ref{dfn:rev_rate}, we state the following achievability bound for general channels:

\begin{lemma} \label{lem:ach_general}
    Given an arbitrary graph $G$, let $|\tilde{f}|$ be the maxflow formed when the capacity of every edge $e$ is set to $\tilde{\rate}_{P_e}$. Then
    \begin{equation}
    \rate_G \geq |\tilde{f}|
    \end{equation}
\end{lemma}

\begin{proof}
We use the same proof as the pairwise reversible case. The only use of the pairwise reversibility condition is in Lemma \ref{lem:separated_inputs} when constructing the codewords $\vec{x}^{(j)}_1, \vec{x}^{(j)}_2, \ldots, \vec{x}^{(j)}_M$ of length $\ell = M!$ such that $\db(\vec{x}^{(j)}_{m_1}, \vec{x}^{(j)}_{m_2}, (P^{(j)})^\ell) \geq \ell \cdot \rate_{M,P^{(i)}}$. However, such codewords of length $\ell = M!$ with $\db(\vec{x}^{(j)}_{m_1}, \vec{x}^{(j)}_{m_2}, (P^{(j)})^\ell) \geq \ell \cdot \tilde{\rate}_{M,P^{(i)}}$ (i.e., with $\tilde{\rate}$ in place of $\rate$) are still guaranteed to exist by Lemma \ref{lem:db_codebook} even when the pairwise reversibility condition is dropped. Outside of Lemma \ref{lem:separated_inputs}, the proof follows with $\rate_{M,P^{(i)}}$ replaced by $\tilde{\rate}_{M,P^{(i)}}$ where appropriate.
\end{proof}

\subsection{Zero-Rate Error Exponents} \label{sec:zero_rate_proof}

In this subsection, we prove Theorem \ref{thm:main_zero_rate}. 
 We start by relating $\tilde{\rate}_{M,P}$ to the zero-rate error exponent:

\begin{lemma}
    For any $M\geq 2$, we have
    \begin{equation}
    \tilde{\rate}_{M,P} \geq E_P(0) \geq \frac{M-1}{M} \tilde{\rate}_{M,P}
    \label{eq:zero_rate_tilde_r}.
\end{equation}
By taking $M\rightarrow \infty$ on both sides, we also have
\begin{equation}
\lim_{M\rightarrow\infty} \tilde{\rate}_{M,P} = E_P(0).
\label{eq:zero_rate_limit}
\end{equation}
\label{lem:zero_rate_tilde_r}
\end{lemma}
\begin{proof}
We use the expression for $E_P(0)$ in Theorem \ref{thm:berlekamp_zero_rate}.
Let $q$ maximize the expression in \eqref{eq:berlekamp_zero_rate}. Draw $(x_1, x_2,\ldots, x_M)$ independently according to distribution $q$, and observe that the definition of $\tilde{\rate}$ in Definition \ref{dfn:rev_rate} gives
\begin{equation}
\tilde{\rate}_{M,P} \geq \e\left(\frac{2}{M(M-1)} \sum_{1\leq m_1<m_2\leq M} \db(x_{m_1}, x_{m_2}, P)\right) = E_P(0).
\end{equation}
To show the second inequality in the lemma statement, let $q$ be the empirical distribution (type) of $(x_1, x_2,\ldots, x_M)$ achieving \eqref{eq:rev_rate}. Substituting $q$ into Theorem \ref{thm:berlekamp_zero_rate}, we have
\begin{eqnarray}
E_P(0) &\geq& \sum_{x, x'} q_xq_{x'}\db(x,x',P) \\
& = & \frac{1}{M^2} \sum_{1\leq m_1, m_2 \leq M} \db(x_{m_1}, x_{m_2}, P) \\
& = & \frac{2}{M^2} \sum_{1\leq m_1 < m_2 \leq M} \db(x_{m_1}, x_{m_2}, P) \\ 
& \stackrel{\eqref{eq:rev_rate}}{=} & \frac{M-1}{M} \tilde{\rate}_{M,P}.
\end{eqnarray}
\end{proof}

We now prove Theorem \ref{thm:main_converse} by devising a protocol whose error exponent is arbitrarily close to $|f|$, where the edge weights in $f$ are the respective 1-hop zero-rate error exponents (i.e., $E_{P_e}(0)$ for edge $e$).

 Let $M$ be a (large) constant, and let $|\tilde{f}_M|$ be the error exponent formed by setting the capacity of each edge $e$ to $\tilde{\rate}_{M,P_e}$. Since $\tilde{\rate}_{M,P_e} \geq E_{P_e}(0)$ for all $P$, we have $|\tilde{f}_M| \geq |f|$.

Building on Sections \ref{sec:series} and \ref{sec:proof_main}, choose a (large) block size $B$, and generate a sequential block transmission protocol of length $B$ with $M$ messages. We can treat a single sequential block transmission as a discrete memoryless channel $Q_{B,M}$. More precisely, the input to $Q_{B,M}$ is the message that the source node receives, and the output of $Q_{B,M}$ is what the destination receives from a single block. The analysis in Sections \ref{sec:series} and \ref{sec:proof_main} shows that given any $\epsilon>0$, there is some sufficiently large $B$ such that
\begin{equation}
\db(m_1, m_2, Q_{B,M}) \geq B(|\tilde{f}_M| - \epsilon)
\end{equation}
for all $m_1 \ne m_2$.  Substituting into \eqref{dfn:rev_rate}, it follows that
\begin{equation}
 \tilde{\rate}_{M,Q_{B,M}} \geq B(|\tilde{f}_M| - \epsilon).
\end{equation}

We now proceed to bound the zero-rate error exponent of $Q_B$. By Lemma \ref{lem:zero_rate_tilde_r},
\begin{equation}
E_{Q_{B,M}}(0) \geq \frac{M-1}{M} \tilde{\rate}_{M,Q_B} \geq \frac{M-1}{M} B(|\tilde{f}_M| - \epsilon) \geq \frac{M-1}{M} B(|f| - \epsilon)
\end{equation}
Similar to the proof of Corollary \ref{cor:series}, we can run $n/B-|V|$ 
conditionally independent copies of the block protocol in $n$ time steps, which is equivalent to using channel $Q_B$ for a total of $n/B-|V|$ times.  By adopting any coding strategy for $Q_B$ that attains its zero-rate error exponent,\footnote{Such a strategy may be very complex, but we recall that our focus is on information-theoretic limits and not practical protocols.} we deduce that 
\begin{equation}
E_G(0) \geq \lim_{n \rightarrow \infty} -\frac{1}{n} \cdot \frac{n-|V|B}{B} \cdot \frac{M-1}{M} \cdot B(|f| - \epsilon) = \frac{M-1}{M} \cdot (|f|-\epsilon).
\end{equation}
Since $\epsilon$ is arbitrarily small and $M$ is arbitrarily large, we conclude that
\begin{equation}
E_G(0) \geq |f|
\end{equation}
as required.

\section{Further Results} \label{sec:further}

In this section, we provide some additional results that address when it is (or is not) possible to relax the assumptions made in our main results.

\subsection{Converse Bound for \texorpdfstring{$M > 2$}{M>2}}

Recall that $|f_{M,G}|$ denotes the maxflow for the graph $G$ with $M$ messages.  One might hope that the proof of Theorem \ref{thm:main_converse} can be modified to prove a converse bound of $\rate_{M,G} \leq |f_{M,G}|$ (instead of $|f_{2,G}|$). However, there are at least two problems in doing so:
\begin{itemize}
    \item We cannot guarantee that the 1-hop feedback and non-feedback error exponents are equal when $M > 2$;
    \item The required analog of Lemma \ref{lem:parallel} may not hold when $M=2$.
\end{itemize}
In this subsection, we partially address these issues by imposing additional assumptions. Specifically, we prove the following:
\begin{theorem}
    Suppose that the following conditions are satisfied:
    \begin{itemize}
        \item $G$ has a mincut with no back-edges (i.e., all edges are from the source side to the destination side);
        \item Every channel in $G$ is pairwise reversible.
    \end{itemize} 
    Then, we have $\rate_{M,G} = |f_{M,G}|$.
    \label{thm:converse_other}
\end{theorem}

We proceed with the proof. 
Although the bound in Lemma \ref{lem:cut_bound} uses the feedback error exponent, we do not need to consider feedback if the partition has no back-edges. We immediately deduce the following:
\begin{lemma}
Consider any partition $(A,B)$ of the nodes in $G$ with the source node in $A$ and the destination node in $B$. Let $P^{(1)}, P^{(2)},\ldots, P^{(k)}$ be the channels associated with edges from $A$ to $B$, and let $\hat{P} = P^{(1)}\times P^{(2)} \times \ldots \times P^{(k)}$ be the product channel. If there are no back-edges, then for all $M$, we have
\begin{equation}
\rate_{M,G} \leq \rate_{M,\hat{P}}.
\end{equation}
\label{lem:cut_bound_no_back_edge}
\end{lemma}

Starting from a mincut with no back-edges, let $P^{(1)}, P^{(2)},\ldots, P^{(k)}$ be the channels going across the cut. Due to our pairwise reversible assumption, we may apply Lemma \ref{lem:prod_parallel_reversible} to obtain
\begin{equation}
\rate_{M,\hat{P}} = \sum_i \rate_{M, P^{(i)}} = |f_{M,G}|, \label{eq:rate_match}
\end{equation}
where the last step holds by maxflow-mincut duality.  This completes the converse part of Theorem \ref{thm:converse_other}. The achievability part is an immediate consequence of Theorem \ref{thm:main}.

Next, for zero-rate error exponents, we show that the pairwise reversibility condition can be dropped:
\begin{lemma}
    Suppose that $G$ has a mincut with no back-edges.
    Then, we have $E_G(0) = |f|$ where $|f|$ is the maxflow formed by setting the capacity of every edge $e$ to $E_{P_{e}}(0)$.
    \label{thm:converse_other_zero_rate}
\end{lemma}
\begin{proof}
Similar to \eqref{eq:rate_match}, the proof of Lemma \ref{lem:prod_parallel_reversible} shows that
    \begin{equation}
        \tilde{\rate}_{M,\hat{P}} = \sum_i \tilde{\rate}_{M,P^{(i)}}.
    \end{equation}
    Taking $M\rightarrow\infty$ and applying \eqref{eq:zero_rate_limit}, we obtain
    \begin{equation}
        E_{\hat{P}}(0) = \sum_{i} E_{P^{(i)}}(0).
    \end{equation}
    Under the same conditions as Lemma \ref{lem:cut_bound_no_back_edge}, we can similarly conclude that $E_{G}(0) \leq E_{\hat{P}}(0)$, and therefore
    \begin{equation}
        E_{G}(0) \leq \sum_{i} E_{P^{(i)}}(0),
    \end{equation}
    which completes the converse part. The achievability part is an immediate consequence of Theorem \ref{thm:main_zero_rate}.
\end{proof}

\subsection{Approximation Guarantees}

Lemma \ref{lem:ach_general} and Theorem \ref{thm:main_converse} reveal that for any network $G$ and any $M$, we have
\begin{equation}
    |\tilde{f}_{M,G}| \leq \rate_{M,G} \leq |f_{2,G}|,
\label{eq:rate_bound}
\end{equation}
where we recall that $\tilde{f}_{M,G}$ uses \eqref{eq:rev_rate} for the edge weights in the maxflow calculation.  
We will proceed to show that this is in fact a 4-approximation:
\begin{lemma} For any $M$ and $G$, we have
\begin{equation}
    |f_{2,G}| \leq 4 {|\tilde{f}_{M,G}|}.
    \label{eq:4approx}
\end{equation}
\label{lem:4approx}
\end{lemma}

\begin{proof}
    It is sufficient to show that $\rate_{2,P} \leq 4\tilde{\rate}_{M,P}$, as this means that $f_{2,G}$ and $\tilde{f}_{M,G}$ are defined with respect to edge weights that differ by at most a factor of 4.  We have 
    \begin{eqnarray}
        \rate_{2,P} &=& \max_{x, x'} \dc(x, x', P) \label{eq:rev_step1} \\
        &=& \max_{x, x'} \max_{s \in [0,1]} \dc(x, x', P, s) \label{eq:rev_step2} \\
        &\leq& 2 \max_{x, x'} \dc(x, x', P, 1/2) \label{eq:rev_step3} \\
        &=& 2 \tilde{\rate}_{2,P}, \label{eq:rev_step4}
        \label{eq:rev_rate_factor_2}
    \end{eqnarray}
    where \eqref{eq:rev_step1} follows from \eqref{eq:two_codewords}, \eqref{eq:rev_step3} follows by using the concavity of $\dc$ in $s$ \cite[Thm.~5]{berlekampI} and considering an equal combination of the values at $s$ and $1-s$ (along with $\dc \ge 0$ for the latter), and \eqref{eq:rev_step4} follows from \eqref{eq:rev_rate}.
    
    Next, let $(x,x')$ be chosen to maximize $\db(x, x', P, 1/2)$. By considering the distribution assigning a mass of $1/2$ to $x$ and $1/2$ to $x'$, we have
    \begin{equation}
        \tilde{\rate}_{M,P} \geq E_P(0) \geq \frac12 \db(x, x', P, 1/2) = \frac12 \tilde{\rate}_{2,P},
        \label{eq:general_loss}
    \end{equation}
    where we first applied Lemma \ref{lem:zero_rate_tilde_r},
    followed by \eqref{eq:berlekamp_zero_rate} with the maximum replaced by the above-mentioned distribution, then \eqref{eq:rev_rate} with $M=2$. 
    Combining equations \eqref{eq:rev_rate_factor_2} and \eqref{eq:general_loss} gives Lemma \ref{lem:4approx}.
\end{proof}

The next lemma gives conditions under which a tighter approximation is possible:
\begin{lemma} \label{lem:2approx}
    Suppose that either of the following conditions hold:
    \begin{itemize}
        \item $M=2$; or
        \item Every channel in $G$ is pairwise reversible.
    \end{itemize}
    Then, we have a 2-approximation in \eqref{eq:rate_bound}, i.e., $|f_{2,G}| \leq 2 {|\tilde{f}_{M,G}|}$.
\end{lemma}
\begin{proof}
    The case $M=2$ follows immediately from \eqref{eq:rev_rate_factor_2}.  Moreover, if every channel is pairwise reversible, then $\tilde{\rate}_{2,P} = \rate_{2,P}$, so the 2-approximation result follows from \eqref{eq:general_loss}.
\end{proof}

We also identify another scenario where our bounds are exactly tight. Consider the $K$-ary symmetric channel with inputs $\{1,2,\ldots, K\}$, outputs  $\{1,2,\ldots, K\}$, and
\begin{equation}
    P(y|x) = \begin{cases}
        1-(K-1)p & y=x\\
        p & y\neq x
    \end{cases}
\end{equation}
for some $p \in \big(0,\frac{1}{K-1}\big)$.

\begin{lemma}
    Suppose that every channel is a $K$-ary symmetric channel for some $K \geq M$ (different channels can have different values of $K$ and different noise levels).  Then $|\tilde{f}_{M,G}| = |f_{2,G}|$, and the protocol given in Section \ref{sec:proof_main} provides the achievability part.
    \label{lem:k_sym}
\end{lemma}

\begin{proof}
    Similarly to the earlier results in this subsection, we simply need to show that $\tilde{\rate}_{M, P_e} = \rate_{2,P_e}$ for each $P_e$. Recall the definition of $\tilde{\rate}$:
    \begin{equation}
        \tilde{\rate}_{M, P_e} = \frac{2}{M(M-1)} \max_{(x_1, \ldots, x_M)} \sum_{1\leq m_1 \leq m_2 \leq K} \db(x_{m_1}, x_{m_2}, P_e). \label{eq:e_repeated}
    \end{equation}
    For the $K$-ary symmetric channel, $\db(x,x',P_e)$ ie either 0 (if $x=x'$) or a fixed positive value (if $x \ne x'$).  
    Since $K \ge M$, we can achieve the maximum in \eqref{eq:e_repeated} by letting $x_1, \ldots, x_M$ be distinct, yielding
    \begin{equation}
        \tilde{\rate}_{M, P_e} = -\log(2\sqrt{(1-(M-1)p)\cdot p} + (M-2)p).
    \end{equation}
    The computation for $\rate_{2,P_e}$ is also straightforward; since $P_e$ is pairwise reversible, we can simply choose any two distinct inputs to obtain
    \begin{equation}
        \rate_{2, P_e} = \max_{x_1, x_2} \db(x_1, x_2, P_e) = -\log(2\sqrt{(1-(M-1)p)\cdot p} + (M-2)p).
    \end{equation}
    Since $\tilde{\rate}_{M, P_e} = \rate_{2, P_e}$ for all edges, it follows that $|\tilde{f}_{M,G}| = |f_{2,G}|$.
\end{proof}

\subsection{Counterexample to a Natural Conjecture}

In Theorem \ref{thm:converse_other}, we assumed that $G$ has a mincut with no back-edges. It might be tempting to believe that this is the case whenever $G$ contains no directed cycles. However, this is false, as we can see in Figure \ref{fig:counterexample}, with the mincut described in the figure caption.

\begin{figure}
\centering
    
\begin{tikzpicture}[scale=2]
\node[draw] at (0 ,0) { $3$ };
\node[draw] at (0 ,1) { $1$ };
\node[draw] at (1 ,0) { $4$ };
\node[draw] at (1, 1) {$2$};
\draw[dashed, ->] (0.15,0) to (0.85,0);
\draw[dashed, ->] (0.15,1) to (0.85,1);
\draw[->] (1,0.85) to (1,0.15);
\draw[->] (0,0.85) to (0,0.15);
\draw[->] (0.85,0.85) to (0.15,0.15);
\end{tikzpicture}
    
    \caption{An example where the only mincut (namely, $A=\{1,3\}$ and $B = \{2,4\}$) contains a back-edge. Solid arrows have infinite capacity and dotted arrows have finite capacity.}
    \label{fig:counterexample}
\end{figure}
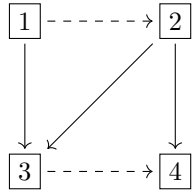

Using the example in Figure \ref{fig:counterexample}, we show that even when every channel is pairwise reversible and the graph is acyclic, it is possible to achieve an error exponent greater than $|f_{M,G}|$ when $M=3$:

\begin{theorem} \label{thm:counter}
    Even if we assume that $G$ is acyclic and every channel is pairwise reversible, there exist scenarios in which $\rate_{M,G} > |f_{M,G}|$ when $M=3$.  Moreover, the same statement holds even when the 1-hop error exponents defining $f_{M,G}$ are replaced by their counterparts with full feedback.
\end{theorem}
\begin{proof}
Consider the graph in Figure \ref{fig:counterexample}. Set $M=3$, and consider the setup described as follows for some $p \in \big(0,\frac{1}{3}\big)$:
\begin{itemize}
\item On the solid edges, infinite noiseless communication is allowed.
\item On the edge $1\rightarrow 2$, we consider a ternary symmetric channel with inputs $\{0,1,2\}$, outputs $\{0,1,2\}$, and with $\p(y|x)=1-2p$ if $y=x$, and $\p(y|x)=p$ otherwise.
\item On the edge $3\rightarrow 4$, we consider a binary symmetric channel with noise parameter $p$.
\end{itemize}
We will prove that as $p$ becomes sufficiently small, it is possible to beat the maxflow rate. For the rest of this section, asymptotic notation refers to the limit as $p\rightarrow 0$.

Let us first compute the error exponents of all the channels. Since all channels are pairwise reversible, we may compute the error exponents using $\rate_{M,P} = \tilde{\rate}_{M,P}$ as per \eqref{eq:rev_rate}. For the ternary symmetric channel, \eqref{eq:rev_rate} is maximized when $x_1, x_2, x_3$ are all distinct, so the error exponent is given by
\begin{equation}
    \rate = -\log (2\sqrt{p(1-2p)} + p) = \frac12 \log\frac1p + \calO(1).
\end{equation}
For the binary symmetric channel, only two of $x_1,x_2,x_3$ can be distinct, so the error exponent is given by
\begin{equation}
    \rate = -\frac23 \log(2\sqrt{p(1-p)}) = \frac13 \log\frac1p + \calO(1).
\end{equation}
Hence, the maxflow is simply the sum of the two error exponents, i.e., $\frac56 \log\frac1p + \calO(1)$. 

We will prove that $\log\frac1p+\calO(1)$ is an achievable error exponent, thus disproving the reasonable conjecture that the maxflow is an upper bound to the error exponent. 
Similar to Section \ref{sec:series}, define a single ``sequential transmission'' as follows:
\begin{itemize}
    \item Node 1 inputs the true message $m$ into its channel connecting to node 2, and also sends $m$ to node 3;
    \item Node 2 then tells both nodes 3 and 4 what it received;
    \item Node 3 then looks at both the ground truth from node 1 and the forwarded value from node 2. If they are the same, message $0$ is sent to node 4, and otherwise, message $1$ is sent to node 4.
\end{itemize}
A total of $n-3$ such transmissions can occur in the $n$ time steps by chaining them one after the other (similar to Figure \ref{fig:block-protocol}).  Moreover, we can view a single sequential transmission as a discrete memoryless channel $Q$, where the input is $m$, and the output consists of what node 4 receives from both nodes 2 and 3.

We can classify the result of each sequential transmission into three possible outcomes:
\begin{itemize}
    \item[(i)] Node 4 receives message 1 from node 3.  This means that either node 2 received the wrong message (and therefore node 3 sent 1), or the message from node 3 to node 4 got flipped. Each event happens with probability $\calO(p)$, so the total probability of this is $\calO(p)$.
    \item[(ii)] Node 4 receives message 0 from node 3, and receives the wrong message from node 2. In this case, both node the messages from 1 to 2 and 3 to 4 must be corrupted, so this happens with probability $\calO(p^2)$
    \item[(iii)] Node 4 receives message 0 from node 3, and receives the correct message from node 2. This occurs with probability $1-\calO(p)$.
\end{itemize}
\end{proof}

We summarize the end-to-end transition probabilities in Figure \ref{fig:transition}:
\begin{figure}[ht]
    \centering
    \begin{tabular}{llllll}
    &   & \multicolumn{4}{l}{\bf output}  \\
    &   & (0,0) & (1,0) & (2,0) & (*,1) \\
    \multirow{3}{*}{\bf input} & 0 & $1-\calO(p)$ & $\calO(p^2)$ & $\calO(p^2)$ & $\calO(p)$  \\
    & 1 &$\calO(p^2)$ & $1-\calO(p)$ & $\calO(p^2)$ & $\calO(p)$  \\
    & 2 & $\calO(p^2)$ & $\calO(p^2)$ & $1-\calO(p)$ & $\calO(p)$ 
    \end{tabular} 
    \medskip
    
    \caption{The input represents the message received by node 0. An output of $(a,b)$ indicates that $a$ was received from node 2 and $b$ was received from node 3. $(*,1)$ indicates that 1 is received from node 3, with arbitrary (ignored) information from node 2.}
    \label{fig:transition}
\end{figure}

Since we are allowed $n-3$ uses of $Q$ in $n$ time steps, we have $\rate_{3,G} \geq \rate_{3,Q}$. We now proceed to bound $\rate_{3,Q}$; we have
\begin{equation}
    \db(0,1,Q) = -\log\sum_{y}\sqrt{Q_0(y)Q_1(y)} = -\log \sqrt{\calO(p^2)} = \log\frac1p + \calO(1)
\end{equation}
and similarly when $(0,1)$ is replaced by any two distinct symbols.  Using Lemma \ref{lem:opt_rev}, this gives
\begin{equation}
    \rate_{3,Q} \ge \log\frac1p + O(1),
\end{equation}
showing that an error exponent of $\log\frac1p+O(1)$ is achievable.

This same counterexample in fact also establishes the second part of Theorem \ref{thm:counter} in which the mincut is defined with respect to edge capacities $\rate^f_{P_e}$ instead of $\rate_{P_e}$.  
To see this, first note that for the ternary symmetric channel, we have $\rate_3 = \rate_2$; this follows from the proof of Lemma \ref{lem:k_sym}, with the ternary symmetric channel being pairwise reversible. On the other hand, $\rate_2 = \rate^f_2$ (Theorem \ref{thm:feedback_m_2}), $\rate^f_2 \geq \rate^f_3$ (error exponent is decreasing in $M$) and $\rate^f_3 \geq \rate_3$ (feedback cannot decrease error exponents). Therefore, the ternary symmetric channel has $\rate_2 = \rate^f_2 \geq \rate^f_3 \geq \rate_3 = \rate_2$; this means that all the inequalities must hold with equality, yielding
\begin{equation}
    \rate^f_3 = \rate_3 = \rate_2 = -\log(2\sqrt{p(1-2p)} + \sqrt{p\cdot p}) = \frac12\log\frac1p+\calO(1).
\end{equation}

For the binary symmetric channel, the feedback error exponent for $M=3$ is given as follows \cite[o.~64]{berlekamp_feedback}:
\begin{equation}
    \rate^f_3 = -\log(p^{1/3}(1-p)^{2/3} + p^{2/3}(1-p)^{1/3}) = \frac13 \log\frac1p + \calO(1),
\end{equation}
which gives a maxflow of $\frac56 \log\frac1p+\calO(1)$.  Once again, since this is strictly below $\log\frac1p + \calO(1)$ for small enough $p$, we obtain the desired result.
\bibliographystyle{IEEEtran}
\bibliography{general}
\end{document}